\documentclass[]{interact}

\usepackage{epstopdf}
\usepackage[caption=false]{subfig}

\usepackage[natbibapa,nodoi]{apacite}
\theoremstyle{plain}
\newtheorem{theorem}{Theorem}[section]
\newtheorem{lemma}[theorem]{Lemma}
\newtheorem{corollary}[theorem]{Corollary}

\theoremstyle{definition}
\newtheorem{definition}[theorem]{Definition}
\newtheorem{example}[theorem]{Example}

\theoremstyle{remark}

\usepackage{listings}%
\usepackage{hyperref}
\usepackage{fancyvrb}
\usepackage{xcolor}

\newcommand{\rio}{\emph{rio}}
\newcommand{\head}{\mathsf{head}}
\newcommand{\body}{\mathsf{body}}
\newcommand{\heads}{\mathsf{heads}}
\newcommand{\bodies}{\mathsf{bodies}}

\newcommand{\cn}{\mathit{Cn}}
\newcommand{\out}{\mathit{out}}
\newcommand{\unsat}{\mathsf{unsat}}
\newcommand{\sat}{\mathsf{sat}}
\newcommand{\wo}{\mathrm{wo}}
\newcommand{\problem}{{{output}}}
\newcommand{\gout}{\mathit{gout}}

\begin{document}

\articletype{Research Article}

\title{A Reduction of Input/Output Logics to SAT}

\author{
\name{A. Steen\textsuperscript{a}\thanks{CONTACT A. Steen. Email: alexander.steen@uni-greifswald.de}}
\affil{\textsuperscript{a}University of Greifswald, Walther-Rathenau-Straße 47, 17489 Greifswald, Germany}
}

\maketitle

\begin{abstract}
Deontic logics are formalisms for reasoning over norms, obligations, permissions and prohibitions.
Input/Output (I/O) Logics are a particular family of so-called norm-based deontic logics that 
formalize conditional norms outside of the underlying object logic language,
where conditional norms do not carry a truth-value themselves.
In this paper, an automation approach for I/O logics is presented that makes use
of suitable reductions to (sequences of) propositional satisfiability problems.
A prototypical implementation, named
\emph{rio} (\textbf{r}easoner for \textbf{i}nput/\textbf{o}utput logics),
of the proposed procedures is presented and applied to illustrative examples.
\end{abstract}

\begin{keywords}
automated reasoning; normative reasoning; input/output logic; non-monotonic reasoning; reduction to SAT
\end{keywords}

\section{Introduction}
Representation of and reasoning with norms, obligations, permissions and prohibitions 
is addressed by the field of deontic logic~\citep{gabbay2013handbook}. It studies reasoning patterns and logical properties
that are not suitably captured by classical propositional or first-order logic.
Various logic formalisms have been proposed to handle deontic and normative reasoning,
including systems based on modal logics~\citep{wright}, dyadic deontic logic~\citep{gabbay2013handbook},
and norm-based systems~\citep{hansen2014reasoning}.
These systems differ in the properties of the obligation operator, and
in their ability to consistently handle deontic paradoxes and/or norm
conflicts~\citep{gabbay2013handbook}.

Input/Output (I/O) logics~\citep{DBLP:journals/jphil/MakinsonT00} 
are a particular norm-based family of systems in which
conditional norms are represented by pairs of formulas. The pairs do not carry truth-values themselves.
I/O logics use an operational semantics based on detachment and come with a variety of
different systems, formalized by different so-called \emph{output operators}.
Given a set of conditional norms $N$, and a set of formulas describing the situational context $A$,
output operators produce a set of formulas that represent the obligations that are in force for that context.
In order to check whether some state of affairs $\varphi$ is obligatory, it suffices to check whether
$\varphi \in \out(N,A)$, where $\out$ is some output operator. 

Unconstrained I/O logics are monotone
and cannot consistently handle norm conflicts (i.e., situations in which norms with conflicting obligations are in force)
without entailing arbitrary conclusions (including falsehood).
A generalization of I/O logics, denoted \emph{constrained I/O logics}~\citep{DBLP:journals/jphil/MakinsonT01},
addresses these shortcomings and provides a framework for non-monotonic normative reasoning.
There are strong connections to default logics, conditional logics and logics for counterfactuals.
For instance, Poole systems~\citep{DBLP:journals/ai/Poole88} can be viewed a special case of
constrained I/O logics, and normal Reiter defaults~\citep{DBLP:journals/ai/Reiter80} are closely
related to them~\citep{DBLP:journals/jphil/MakinsonT01}.

The use of I/O logics for legal reasoning has been proposed by Boella
 and van der Torre (\citeyear{boella4constitutive}).
The legal knowledge base DAPRECO
uses a first-order variant of I/O logics and extensively encodes the
\emph{General Data Protection Regulation} (GDPR) of the European Union~\citep{robaldo2019formalizing}.
Automated reasoning
within knowledge bases such as DAPRECO constitutes a prototypical application scenario of computer-assisted
normative reasoning and is of high interest to both academic and industrial stakeholders.
However, currently there exist only few automated reasoning procedures 
for I/O logics (see related work below).

The decision problem whether $\varphi \in \out(N,A)$ is referred to as \emph{entailment
problem}. This is the decision problem studied in the literature.
In this article a generalization called the \emph{\problem\ problem} is studied: Given a set of norms $N$ and an input $A$,
compute a finite representation $B$ of formulas (called a basis) such that $\cn(B) = \out(N,A)$.
In other words, the entailment problems ask whether a given formula can be detached in a situation
under a set of norms, and the \problem\ problem generates (a finite representation of) the set of \emph{all entailed} formulas.
The entailment problem for $\varphi$ can easily be solved once a basis $B$ of the output is found by checking
whether $B \vdash \varphi$ (using some standard deduction method). The other way around is, however,
not as clear and remains an open question.\footnote{In general, the output set $\out(N,A)$ contains formulas that are themselves not 
occurring in any norm in $N$, so a simple syntactical search is not sufficient.}

One research hypothesis of this work is that both the entailment problem and the \problem\ problem
are useful in practice and somewhat orthogonal: If an agent faces a new situation, a list of obligations
can be generated to guide the agent's next actions. If, on the other hand, only the entailment problem
is available, the agent would need to guess for most plausible obligations in force, and check their 
entailment one after one, possibly missing an important one. After the initial situation has been assessed
by the agent, entailment procedures seem most useful to check for compliance of potential individual
follow-up actions.

This work aims to provide automation for different I/O logic formalisms wrt.\ the \problem\ problem.
To that end, a reduction of I/O logic reasoning to the satisfiability decision problem (SAT)
in propositional logic~\citep{DBLP:series/faia/336} is given and an implementation of the resulting procedures within the reasoner \rio\
is presented.

The contributions of this work are as follows:
  \begin{itemize}
    \item A novel SAT-based reduction for unconstrained and constrained I/O logic
          reasoning is presented. Overall, I/O logic reasoning with eight output operators
          is supported, as well as 16 constrained I/O logic reasoning set-ups.\footnote{
            Some these set-ups are semantically equivalent, see Sect.\ \ref{sec:iol} for a discussion.}
          Up to the author's knowledge, this work constitutes the first automated reasoning
          procedure for unconstrained and constrained I/O logics wrt. the \problem\ problem.
    \item Previous work of the author~\citep{steenDEON} is simplified and generalized so that a finite representation
          of the whole output set is computed (see below for a brief discussion on this).
          This allows the simple automation of more general output operators,
          discussed in the literature~\citep{DBLP:conf/deon/ParentT14}.
    \item A knowledge representation format for automated reasoning in I/O logics is presented.
          This format is aligned with the well-established TPTP standard~\citep{Sut17} for automated reasoning systems, in
          particular it instantiates the TPTP extension to non-classical logics formalisms~\citep{DBLP:conf/paar/SteenFGSB22,DBLP:conf/paar/SteenS24}.
    \item An open-source implementation of the presented reasoning procedures is presented. The automated
          reasoning system is called \rio~\citep{rio}, and reads problems and outputs results according to the TPTP standard.
  \end{itemize}

In Sect.\ \ref{sec:iol} the basic notions of I/O logics are introduced. Sect.\ \ref{sec:reduction}
gives reductions of the different output operators to SAT, and Sect.\ \ref{sec:impl}
presents an implementation of the SAT-based reasoning procedures. In Sect.~\ref{sec:casestudy}
the automation methods are applied to a small case study. Finally,
Sect.\ \ref{sec:conclusion} concludes and sketches further work.

\paragraph*{Related work.}
The idea of reducing entailment in I/O logics to classical entailment is not new.
A so-called \emph{relabelling technique} that reduces the entailment
problem for $\out_2$ and $\out_4$ to classical logic entailments is presented by 
\cite{DBLP:journals/jphil/MakinsonT00} in their original I/O logic article.
While this original reduction provides a more direct (and efficient) approach
for the entailment problem for these output operators, 
it is unclear how to generalize to constrained I/O reasoning, or for application
to the \problem\ problem.
The relabelling approach for the entailment problem
was extended to $\out_1$ and $\out_3$ by \cite{DBLP:conf/kr/CiabattoniR23}.

\cite{DBLP:conf/kr/CiabattoniR23} present proof-search-oriented sequent calculi for so-called causal
input/output logics, and then present a simple connection to the original I/O logics via an additional
entailment check in classical logic. This work focuses on the entailment problem and only on unconstrained I/O logics.
\cite{DBLP:conf/tableaux/Lellmann21} presents sequent calculi aligning the entailment problem in I/O logics
to a fragment of conditional logics. This approach was implemented in Prolog.

Benzmüller et al. (\citeyear{J46}) present a shallow semantical embedding for two unconstrained I/O logic operations
into classical higher-order logic. They make use of the fact that the output operators
$\out_2$ and $\out_4$, cf. Sect.\ \ref{sec:iol}, can be
represented by adequate translations to modal logics \textbf{K} and \textbf{KT}.
Such an approach is not available
for the remaining output operators as well as all for the constrained I/O logics setting.

An oracle-based non-deterministic algorithm for the entailment problem in constrained I/O logic
is presented by~\cite{DBLP:journals/japll/SunR17}. The author is not aware of any implementation, though.

Robaldo (\citeyear{DBLP:conf/icail/Robaldo21}) models reified I/O logics~\citep{DBLP:journals/logcom/RobaldoS17},
an enriched version of the formalism used in this paper,
in the recent W3C Shapes Constraint Language (SHACL)~\citep{SHACL} for reasoning with the W3C RDFs/OWL
standard language for the Semantic Web~\citep{RDF}.

Earlier work of the author addresses ''single-shot'' reasoning procedures for unconstrained
I/O logics~\citep{steenDEON} which seem unfit for generalization to the constrained case. 
In particular, these reasoning procedures only allow for the computation whether a given formula $\varphi$
is output by an (unconstrained) output operator. The generalized approach presented here
computes a finite representation of the infinite output set.

A framework for flexible automated normative reasoning, bridging between the well-established normative and legal knowledge representation standard
LegalRuleML~\citep{DBLP:conf/rweb/AthanGPPW15} and the TPTP standard~\citep{Sut17} for automated theorem provers,
has been developed by Steen and Fuenmayor (\citeyear{SteenFuenmayorRuleML}). The automation procedures
for I/O logics presented here can be seamlessly integrated into that framework, enriching it by I/O logic reasoning capacities.

The author is not aware of any related work wrt.\ the \problem\ problem as defined further above, and 
on automation of constrained I/O logics.

\section{Input/Output Logic \label{sec:iol}}

I/O logics were introduced by Makinson and van der Torre (\citeyear{DBLP:journals/jphil/MakinsonT00,DBLP:journals/jphil/MakinsonT01}) as a formalism for the abstract study
of conditional norms, e.g., obligations constituted by some
normative code.
In the I/O formalism, inputs are mapped to outputs with respect to a normative system
using so-called \emph{output operators}.
Intuitively, inputs describe the situational context and the outputs are factual obligations
in the given context.
The output operators differ in their detachment properties.

Two I/O logics settings are considered:
Every operator's output can be used as-is; this is also referred to as \emph{unconstrained}
I/O logics~\citep{DBLP:journals/jphil/MakinsonT00}. In constrained I/O logics, instead,
additional constraints may be imposed on the output such that it is guaranteed to be consistent
with the constraints~\citep{DBLP:journals/jphil/MakinsonT01}. Constrained I/O logics are non-monotonic,
and allow for defeasible reasoning with conflicting obligations (deontic dilemmas) as well as 
norm violations (contrary-to-duty scenarios).

The most important notions are summarized in the following, with material from the relevant literature~\citep{iologichandbook}
and partly adopted from earlier work of the author~\cite{steenDEON}. Since the I/O logics formalisms are not widely known and might
seem counter-intuitive in the beginning,
the subsequent introduction is supported by a few examples.

\subsection{Unconstrained I/O Logics}\label{ssec:iol}

Let $L$ be the language of classical propositional logic over some set of propositional symbols
 with the usual connectives
$\neg$, $\land$ and $\lor$ for negation, conjunction and disjunction, respectively.
The symbols $\top$ and $\bot$ represent some arbitrary tautology and contradiction, respectively.
Furthermore let $\vdash$ denote the consequence relation of classical propositional logic, i.e., 
for all $x \in L$ and $A \subseteq L$ it holds that $A \vdash x$ if and only if $x$ is a semantical consequence of $A$.
The set of consequences $\cn(\cdot)$ is defined as $\cn(A) := \{ x \in L \mid A \vdash x \}$.

\paragraph*{Syntax.} A \emph{normative system} $N \subseteq L \times L$ is a 
set of pairs $(a,x)$ of formulas.\footnote{
In the literature the formulas are usually referred to as $a, b, c, \ldots$
for inputs and $x, y, z$ for outputs. This is also adopted here.
}
A pair
$(a,x) \in N$ represents the conditional norm that \emph{given $a$, it ought to be $x$}. 
By convention, given a norm $n = (a,x)$ the first element $a$ is also referred to as the \emph{body}
of $n$, denoted $\body(n)$, and the second element $x$ is referred to as the \emph{head} of $n$, 
denoted $\head(n)$. Let furthermore $\heads(N^\prime) := \{ \head(n) \mid n \in N^\prime \}$
for some $N^\prime \subseteq N$, and analogous for $\bodies(N^\prime)$.
If $A$ is a set of formulas, the image of $N$, written as $N(A)$,
is given by $N(A) := \{x \in L \mid (a,x) \in N \textit{ for some $a \in A$} \}$.
Given a normative system $N$ and a set of formulas $A$, $out(N,A) \subseteq L$
denotes the output of input $A$ under $N$ where $out$ is an output operator. 
It is assumed in the following that $A$ and $N$ are finite sets.

\paragraph*{Semantics.}
The semantics of I/O logics is operational in the sense that the meaning of normative concepts
is given by generating outputs with respect to a normative system and an input by detachment.
The four distinct output operators
$\out_i$, $i \in \{1,2,3,4\}$, studied in the literature are defined
as follows~\citep{DBLP:journals/jphil/MakinsonT00}:
\begin{equation*}
\begin{split}
\out_1(N,A) &= \cn(N(\cn(A))) \\
\out_2(N,A) &= \bigcap \big\{ \cn(N(V)) \mid  V \supseteq A, V \text{ complete}\big\} \\
\out_3(N,A) &= \bigcap \big\{ \cn(N(B)) \mid  A \subseteq B = \cn(B) \supseteq N(B) \big\} \\
\out_4(N,A) &= \bigcap \big\{ \cn(N(V)) \mid  A \subseteq V \supseteq N(V), V \text{ complete} \big\}
\end{split}
\end{equation*}
\noindent where $V \subseteq L$ is called complete iff $V = L$ or $V$ is a maximally consistent
set.\footnote{
A clarification might be in order: The set $B$ in the definition of $\out_2$
is defined to be a least superset of $A$ that is closed both under $N(.)$ and $\cn(.)$. This also applies to 
the set $V$ in the definition of $\out_4$, but the definition can be simplified as $V = \cn(V)$ for
every complete set $V$.
} 

$\out_1$, called \emph{simple-minded}, mainly illustrates the basic machinery of
detachment underlying the output operators. $\out_2$ is called \emph{basic} and augments
the simple-minded output with reasoning by cases. $\out_3$ and $\out_4$ are referred to as 
\emph{reusable} and \emph{basic reusable}, respectively, and extend the former operators
with iterative detachment, i.e., the respective output is exhaustively recycled as input.

Unconstrained I/O logic reasoning with $\out_1, \out_2$ and $\out_4$ is shown to be coNP-complete
by Sun and Robaldo \citeyear{DBLP:journals/japll/SunR17}, $\out_3$ was recently shown to be 
coNP-complete as well~\citep{DBLP:conf/kr/CiabattoniR23}.

\begin{example}
Let $N = \big\{(a, x), (b, y), (x \land y, z) \big\}$ and let 
$A_1 = \{ a, b \}$ and $A_2 = \{ a \lor b \}$.
The following outputs are generated using the different output operators:
\begin{enumerate}
  \item $\out_1(N, A_1) = \out_2(N, A_1) = \cn(\{x,y\})$
  \item $\out_3(N, A_1) = \out_4(N, A_1) = \cn(\{x,y,z\})$
  \item $\out_1(N, A_2) = \out_3(N, A_2) = \cn(\emptyset)$
  \item $\out_2(N, A_2) = \cn(\{x \lor y\})$
  \item $\out_4(N, A_2) = \cn(\{x \lor y, z\})$
\end{enumerate}
Note that in every case it holds that $\{a,b\} \nsubseteq \out_{\star}(N, .)$. 
Also, the output set of every I/O operation is infinite as the set of consequences $\cn(.)$ 
contains at least the formulas $\top$, $\top \land \top$, $\top \land \top \land \top$, \ldots, etc.
\end{example}

For every output operator $\out$ a \emph{throughput}-variant of $\out$, denoted $\out^+$,
may be defined by $\out^+(N,A) := \out(N \cup I, A)$ where $I = \{ (x,x) \mid x \in L \}$~\citep{DBLP:journals/jphil/MakinsonT00}.
Throughput operators forward the input $A$ to the output set, in addition to the usual output generated by its
underlying output operator, i.e., if $a \in A$ then $a \in \out^+(N,A)$. It is known that $\out_2^+$ and $\out_4^+$ 
collapse to classical logic~\citep{DBLP:journals/jphil/MakinsonT01}. Together with the four principal output operators
above, there hence exist seven distinct output operations in total.

It is important to note that whereas a conditional norm $(a,x)$ looks similar to an
implication $a \Rightarrow x$, the output semantics is quite different: 
As an example, assume $N = \{(a,x)\}$ and $A = \{a\}$. In this case
all output operators $\out_i$ yield $\out_i(N,A) = \cn(\{ x \})$. In
particular $a \notin \out_i(N, A)$, while an analogous encoding as implication on the object logic level 
entails that
 $\{a, x\} \subseteq \cn(\{a \Rightarrow x, a\})$. 
 An interpretation as material implication can be modeled in I/O logics
 by using, e.g., the output operator $\out_4^+$~\citep{DBLP:journals/jphil/MakinsonT00}.
The operational approach of I/O logics allows for a more fine-grained control over
detached obligations. This is particularly important in the presence of norm violations
where a representation using material implications leads to inconsistencies.\footnote{
Consider the normative system $N = \{(\top, \neg k)\}$, where $k$ stands for \emph{killing someone}.
The norm $(\top, \neg k)$ expresses the (unconditional) obligation not to kill anyone. 
Consider a situation $A$ where this obligation has been violated, i.e., $\{k\} \subseteq A$.
It follows that $\out_i(N, A) = \cn(\neg k)$, but 
$\cn(A \cup \{\top \Rightarrow \neg k\}) = \cn(\{\bot\}) = L$.
}

Makinson and van der Torre (\citeyear{DBLP:journals/jphil/MakinsonT00}) provide a syntactic characterization of the different output operators
in terms of derivation rules. They are not discussed here
as they do not provide effective means for automation, in particular (1) they
do not allow backward proof search (as done in most contemporary automated reasoning systems) and (2)
they only derive the fact that some specific formula is in the output set (as opposed to describing the whole output set
as done in this work).

\subsection{Constrained I/O Logics}
Constrained I/O logic~\citep{DBLP:journals/jphil/MakinsonT01} extends unconstrained I/O logic with
a notion of 
additional constraints that outputs must to be consistent with.
This enables the I/O formalism to consistently handle deontic paradoxes such as norm violations and 
deontic dilemmas~\citep{gabbay2013handbook}.
To that end, a fundamental technique from belief change~\citep{DBLP:series/sbis/FermeH18} is adapted,
as used in, e.g., so-called contraction operations:
The set of norms is restricted just enough that consistency of the output with the constraints 
is guaranteed.
Of course there exist multiple accordingly restricted sets in
general, and each such set potentially generates a different output set.

Let $C$ be a set of formulas called \emph{constraints}.
The notions of \emph{maxfamilies} and \emph{outfamilies} 
wrt.\ some output operator $\out$ are defined as follows:\footnote{ \label{fn:clarification}
  The reference to the output operator $\out$ occurring as subscript in the terms $\mathit{maxfamily}_\out(N, A, C)$
  and $\mathit{outfamily}_\out(N, A, C)$ is usually omitted
  in the literature. In the following, it is explicitly stated as clarification.
}
\begin{equation*}\begin{split}
\mathit{maxfamily}_\out(N, A, C) &:=
  \big\{ N^\prime \subseteq N \mid N^\prime \text{ is $\subseteq$-maximal in } N, \out(N^\prime,A) \cup C \nvdash \bot \big\} \\
  \mathit{outfamily}_\out(N, A, C) &:=
  \big\{ \out(N^\prime,A) \mid N^\prime \in \mathit{maxfamily}_\out(N,A,C) \big\}
\end{split}\end{equation*}
Intuitively, a $\mathit{maxfamily}$ formalizes the concept of maximal choices of norms from $N$ such that output generated
using that set is consistent with $C$.
An $\mathit{outfamily}$ simply maps norm sets to their respective output set, i.e., detached
outputs via $\out$ consistent with $C$. 

\begin{example}[Contrary-to-duty]\label{example:chisholm}
The so-called \emph{Chisholm paradox} is a well-known contrary-to-duty scenario
in the deontic logic community~\citep{chisholm}. Consider three norms 
$N = \{ (\top, h), (h, t), (\neg h, \neg t)\}$ where $h$ and $t$ stand for
\emph{helping your neighbor} and \emph{telling him/her that you come over}, respectively.
The set $N$ represents following norms: \emph{you ought to help your neighbor},
\emph{if you help your neighbor, then you ought to tell him/her
that you are coming over}, and \emph{if you don't help your neighbor,
 then you ought to not tell him/her that you come over}, respectively.
Furthermore assume that $A = \{ \neg h\}$.

For $i \in \{3,4\}$ it holds that $\bot \in \out_i(N,A)$, i.e., unconstrained
I/O logic produces an inconsistent output set.
Now consider the constraint $C = A$.\footnote{
  In the context of norm violations, the constraints represent facts that are undeniably true
  and cannot be changed. The output operator
  is constrained to only output obligations that are consistent with the constraints, e.g., in order to
  compute what other obligations are in force even if others have already been violated.
  For a thorough discussion on the choice of constraints in scenarios of norm violation
  see, e.g.,~\citep{hansson1969analysis}.
} It follows that $\mathit{maxfamily}_\out(N,A,C) = \big\{ \{(h, t), (\neg h, \neg t) \} \big\}$
and hence $\mathit{outfamily}_\out(N,A,C) = \big\{ \{ \neg t \} \big\}$ for $\out \in \{\out_3,\out_4\}$.
\end{example}

In constrained I/O logic, the output is generated from outfamilies (i.e., sets of sets of formulas)
by aggregate functions. Two options are discussed in the literature:
When assuming so-called \emph{credulous} output, denoted $\out_{\cup}^C(N,A)$, the
join of all elements is returned, i.e., 
$\out_{\cup}^C(N,A) := \bigcup\mathit{outfamily}_\out(N,A,C)$.
In \emph{skeptical} output, denoted $\out_{\cap}^C(N,A)$, all common output formulas are returned,
i.e., $\out_{\cap}^C(N,A) := \bigcap\mathit{outfamily}_\out(N,A,C)$.\footnote{
  The remark of footnote~\ref{fn:clarification} also applies to
  the reference to $C$ in the term $\out_{\cdot}^C(N,A)$.
}
The terms credulous and skeptical indicate that the latter variant is more restrictive than the former
(or ''more cautious'' when generating outputs resp.\ obligations).

\begin{example}[Contrary-to-duty, continued]
Let $N$, $A$ and $C$ be as in Example~\ref{example:chisholm}.
For $\out \in \{\out_3, \out_4\}$ it holds that $\out_{\cap}^C(N,A) = \out_{\cup}^C(N,A) = \cn(\{ \neg t \})$, in
particular $\out_{\cdot}^C(N,A) \nvdash \bot$. This means that constrained I/O logic
consistently infers that, in the given situation you ought not to tell your neighbor that you're coming over.
Note that obligations that are in violation with the situation (but, in principle, still in force) are not in the output set.
\end{example}

Constrained I/O logic is non-monotonic and allows for defeasible
reasoning~\citep{DBLP:journals/jphil/MakinsonT01}.
This is an important capability of reasoning formalisms in the context of, among others,
normative reasoning, learning, and multi-agent systems with incomplete knowledge and/or 
fallible perception~\citep{DBLP:series/sbis/FermeH18,gabbay2013handbook}.
Credulous and skeptical constrained I/O logics reasoning with all the introduced output operators
is $\text{NP}^{\text{NP}}$-complete and $\text{coNP}^{\text{NP}}$-complete, respectively~\citep{DBLP:journals/japll/SunR17}
(that is, on the second level of the polynomial hierarchy, $\Sigma^p_2$ and $\Pi^p_2$, respectively).

Constrained I/O logic has no known proof theory~\citep[Sect.\  5.5]{iologichandbook}.


\section{Reductions to SAT \label{sec:reduction}}
Reductions of the \problem\ problem in unconstrained and constrained I/O logic reasoning to a sequence of
SAT problems is given.
The main idea is to represent the (infinite) output set of an I/O logic output operator
by a finite set of representatives.
Let $\Phi$ be a set of formulas. A \emph{finite base} $\mathbf{B} \subseteq L$ of
$\Phi$ is a finite set of formulas such that each formula of $\Phi$ is entailed by it, i.e.,
$\mathbf{B} \vdash \phi$ for each $\phi \in \Phi$.
With the goal of logic automation in mind, the \problem\ problem instances considered
in this work are assumed to be finite in the following sense: The vocabulary of $L$ (i.e., set of
propositional atoms) is finite, and both the input $A$ and the set of norms $N$ are finite.
In this setting, it is guaranteed that finite bases for the (deductively closed) output sets 
always exist.
The construction of the bases is reduced to
procedures that use SAT decisions as subroutines. Once a basis is constructed,
checking whether a formula is entailed by it (and hence in the output set)
can again be decided by a SAT problem.
In the following, $\unsat(\Phi)$ denotes the SAT problem of showing
unsatisfiability of $\Phi \subseteq L$.

\subsection{Reductions for $\out_1$ and $\out_3$}
\newcommand{\triggers}{\mathbin{\triangleright}}
\newcommand{\weaklytriggers}{\mathbin{\blacktriangleright}}
\newcommand{\minweaklytriggers}{\mathbin{\blacktriangleright\!\!\blacktriangleright}}

Output of the operators $\out_1$ and $\out_3$ is given by single-step
and iterative detachment, respectively.
The common detachment machinery is formalized by \emph{directly triggered norms}:
\begin{definition}[Directly Triggered Norm]
Let $n \in N$ be a conditional norm.
Input $A$ \emph{directly triggers} $n$, written $A \triggers n$, if and only if $A \vdash \body(n)$.
\end{definition}

For any $n \in N$ it holds that $A \triggers n$ iff
$\unsat\big(A \cup \{\neg \body(n)\}\big)$. 
The basis of a simple-minded output via $\out_1$,
denoted $\mathbf{B}_1(N, A)$, is then defined as:
\begin{equation*}
  \mathbf{B}_1(N, A) := \big\{ \head(n) \in L \mid A \triggers n \text{ for some } n \in N \big\}
\end{equation*}
It is easy to see that $\mathbf{B}_1(N, A)$ is indeed a finite basis of $\out_1(N,A)$:

\medskip
\begin{theorem}\label{theorem:out1}
$\phi \in \out_1(N,A)$ if and only if $\mathbf{B}_1(N, A) \vdash \phi$.
\end{theorem}
\begin{proof}
It holds that $\mathbf{B}_1(N, A) = N(\cn(A))$. By definition of $\out_1$
the assertion follows directly.
\end{proof}

For reusable output $\out_3$ the basis $\mathbf{B}_3(N, A)$
contains every norm's head that can be detached by iterating direct triggering.
This is similar to calculating the so-called \emph{bulk increments}~\citep{stolpe2008norms},
with the difference that bulk increments are closed under $\cn(.)$ and are hence
not finite.
More formally, let $\mathbf{B}_3^i(N, A)$, $i \geq 0$, be a family of 
sets of formulas:
\begin{equation*}\begin{split}
  \mathbf{B}_3^0(N, A) &:= \mathbf{B}_1(N, A) \\
  \mathbf{B}_3^{i+1}(N, A) &:= \mathbf{B}_1\big(N, A \cup \mathbf{B}_3^{i}(N, A)\big) 
\end{split}\end{equation*}
Because $\vdash$ is monotone the $\mathbf{B}_3^{i}(N, A)$
are totally ordered with respect to set inclusion, i.e.,
$\mathbf{B}_3^i(N, A) \subseteq \mathbf{B}_3^{i+1}(N, A)$ for every $i \geq 0$.
Also $\mathbf{B}_3^i(N, A)$ is bound from above by $\heads(N)$.  
It follows that there exists a maximal element within that chain,
in particular there exists some $j \geq 0$ such that $\mathbf{B}_3^{j}(N, A) = \mathbf{B}_3^{j+1}(N, A)$.
We set $\mathbf{B}_3(N, A) := \mathbf{B}_3^{j}(N, A)$.

\medskip
\begin{theorem}\label{theorem:out3}
$\phi \in \out_3(N,A)$ if and only if $\mathbf{B}_3(N, A) \vdash \phi$.
\end{theorem}
\begin{proof}
Assume $\phi \in \out_3(N,A)$. Since $N(\cdot)$ and $\cn(\cdot)$
are monotone, it holds that $\cn(N(A^*)) \vdash \phi$
where $A^*$ is the least superset of $A$ that is closed
both under $\cn(\cdot)$ and $N(\cdot)$. 
Then, because of compactness and by definition of $\mathbf{B}_3(N, A)$, there exists some $i \geq 0$ such that 
$\cn(A^*) = \cn(\mathbf{B}_3^i(N, A))$ and thus
$\mathbf{B}_3(N, A) \vdash \phi$.

Conversely, if $\mathbf{B}_3(N, A) \vdash \phi$
then $\mathbf{B}_3^i(N, A) \vdash \phi$ for some $i \geq 0$.
So there exists a subset $\{ n_1, \ldots, n_k \} \subseteq N$ of norms 
such that $A \cup \mathbf{B}_3^i(N, A) \vdash \body(n_1) \land \ldots \land \body(n_k)$
and $\{ \head(n_1), \ldots, \head(n_k) \} \vdash \phi$.
It holds that $A \cup \mathbf{B}_3^i(N, A) \subseteq B$ for every set $B$ that is closed
both under $\cn(\cdot)$ and $N(\cdot)$, and hence $\phi \in \out_3(N,A)$.
\end{proof}

\paragraph*{Construction of bases.}
The construction of $\mathbf{B}_1(N, A)$ is done via
the procedure $\texttt{directlyTriggeredNorms}(A,N)$ that lifts the
notion of a directly triggered norm to sets of norms $N$ by filtering as follows:
\begin{equation*}
\texttt{directlyTriggeredNorms}(A, N) := \{ n \in N \mid A \triggers n \}
\end{equation*}
Finally we set
$$\mathbf{B}_1(N, A) := \heads\big(\texttt{directlyTriggeredNorms}(A, N)\big).$$
This amounts to a sequence of $|N|$ SAT problems.

The computational construction of the basis $\mathbf{B}_3(N, A)$ uses
a fixed-point iteration.
For each $i \geq 0$ let $N_i := \{ n \in N \mid A_i \triggers n \}$ be the subset of all norms that are directly triggered by $A_i$,
where $A_i \subseteq L$ is the augmented input after each iteration. Initially, $A_0 := A$
and let $A_{i+1} := A_i \cup \{ \head(n) \mid n \in N_i \}$.
Calculating $N_i$ amounts to at most $|N|$ SAT problems. 
As soon as $N_{j+1} = N_j$ for some $j \geq 0$, a fixed point has been found and
$\mathbf{B}_3(N, A) := \heads(N_j)$. 
This procedure is displayed in Fig.~\ref{fig:basis3}. Here, the incremental sets $A_i$ and $N_i$ are represented
by the variables \texttt{A'} and \texttt{N'}, respectively. The convergence criterion is slightly rewritten
to avoid redundant computations: Firstly, the variable \texttt{T} collects norms that are newly triggered by \texttt{A'}
in the respective iteration. Only newly triggered norms have to be collected in $\texttt{N} \setminus \texttt{N'}$,
as the norms in \texttt{N'} were already triggered before.
If no new norms can be detached (i.e., $\texttt{T'} = \emptyset$) the loop terminates.

The computation of $\mathbf{B}_3(N, A)$ requires $\mathcal{O}(|N|^2)$
individual SAT problems in the worst case: Consider a normative system
$N = \{n_1, \ldots n_k\}$, $k > 1$, where each norm $n_{i+1}$ is directly triggered
by $A \cup \head(n_{i})$, $i \geq 1$. In this case, in every iteration,
$\mathcal{O}(k-i)$ SAT problems are generated, $1 \leq i < k$, finally converging
to a fixed point after $k = |N|$ iterations.

\begin{figure}[tb]
  \centering
  \begin{lstlisting}[frame=lines,basicstyle=\small\ttfamily,mathescape,morekeywords={if,for,endfor,endif,then}]
A' := A  
T  := directlyTriggeredNorms(A, N)  
N' := T  
while $\mathtt{T} \neq \emptyset$ do
  A' := A' $\cup$ $\heads$(T)
  T  := directlyTriggeredNorms(A', N $\setminus$ N')
  N' := N' $\cup$ T
end
return $\heads$(N')
  \end{lstlisting}
  \caption{Construction of base $\mathbf{B}_3(N, A)$ using fixed-point iteration. \label{fig:basis3}}
\end{figure}


\subsection{Reductions for $\out_2$ and $\out_4$}
Finite bases for basic output $\out_2$ are defined on top of a generalized notion of
norm triggering, called \emph{weakly triggered} norms.
Recall that $\out_2$ augments simple-minded output $\out_1$ with
reasoning by cases, which makes norm triggering not only dependent on the input
$A$ but also on the existence of other norms in $N$.
Consider $N = \{ (a,x), (\neg a, x), (b, y) \}$. For any $A$ it holds that
$x \in \out_2(N, A)$ because $a \lor \neg a$ is a tautology and hence, by definition, $x$ is
in the output set.

\begin{definition}[Weakly Triggered Set of Norms]
Let $N$ be a normative system and $N^\prime \subseteq N$ with $N^\prime \neq \emptyset$.
Input $A$ \emph{weakly triggers $N^\prime$}, written $A \weaklytriggers N^\prime$, if and only if 
\begin{equation*}
A \vdash \bigvee_{n \in N^\prime}  \body(n).
\end{equation*}
The \emph{weak output} generated by $N^\prime$ is given by $\mathrm{wo}(N^\prime) := \bigvee_{n \in N^\prime}  \head(n)$.
\end{definition}
Intuitively, if a set of norms $N^\prime$ is weakly triggered by some input $A$, then at least one norm
in $N^\prime$ can indeed be detached by $A$ but it is unknown which one (following the reasoning
by cases approach of basic output). 
In the above example it holds that $A \weaklytriggers \{ (a,x), (\neg a, x) \}$ and
$A \weaklytriggers N$. If $N^\prime \subseteq N$ is weakly triggered by $A$, then so
is any superset $M \supseteq N^\prime$. Moreover, if $A \weaklytriggers N$ and $A \weaklytriggers M$ where
$N \subseteq M$, then $\mathrm{wo}(N) \vdash \mathrm{wo}(M)$. In the example above,
it holds that $\mathrm{wo}\big(\{ (a,x), (\neg a, x) \}\big) = x$ and $\mathrm{wo}(N) = x \lor y$.
In particular $\{x\} \vdash x \lor y$.

Consider another scenario:
Let $A = \{a \lor b \}$ and let $N = \{ (a,x), (b, y), (c, z) \}$.
Neither $a$ nor $b$ is entailed by $A$, but $x \lor y \in \out_2(N,A)$.
In particular, $\out_2(N,A)$ contains elements that are not 
head of any norm from $N$. However, it holds that $A \weaklytriggers \{(a,x), (b, y)\}$,
and $\mathrm{wo}\big(\{(a,x), (b, y)\}\big) = x \lor y$, which is exactly the desired output.

Note that every directly triggered norm $n$ also forms a weakly triggered singleton set $\{ n \}$ by definition.
Furthermore, if the input $A$ is inconsistent then every singleton subset of $N$ is weakly triggered.
In general, there may be many different weakly triggered sets in $N$.

$\mathbf{B}_2(N, A)$ is then defined as follows:
\begin{equation*}
  \mathbf{B}_2(N, A) := \big\{ \mathrm{wo}(N^\prime) \mid A \weaklytriggers N^\prime \text{ for some } N^\prime \subseteq N \big\}
\end{equation*}

\begin{theorem}\label{theorem:out2}
$\varphi \in \out_2(N,A)$ if and only if $\mathbf{B}_2(N, A) \vdash \varphi$.
\end{theorem}

\begin{proof}
Assume $\varphi \in \out_2(N,A)$. For every complete extension $V \supseteq A$ of $A$ 
there exist norms $N^V = \{n^V_1, \ldots, n^V_{k_V} \} \subseteq N$, $1 \leq k_V \leq |N|$, such that
$V \vdash \bigwedge_{1 \leq i \leq k_V} \body(n^V_i)$
and $\bigwedge_{1 \leq i \leq k_V} \head(n^V_i) \vdash \varphi$.
Then
\begin{equation} \label{eq:1}
  \bigvee_{\substack{V \supseteq A,\\ V \text{complete}}} \bigwedge_{1 \leq i \leq k_V} \body(n^V_i)
\end{equation}
is entailed by all complete extensions $V \supseteq A$.
Also, $A$ entails (\ref{eq:1}) because every interpretation that satisfies $A$ also satisfies some complete extension $V \supseteq A$,
and by construction each complete extension entails at least one disjunct of (\ref{eq:1}).
By exhaustive rewriting an equivalent formula in conjunctive normal form is obtained:
\begin{equation} \label{eq:2}
  \bigwedge_{1 \leq i \leq j} \body(n_{i,1}) \lor \ldots \lor \body(n_{i,v}),
\end{equation}
where $v = \lvert\{ V \supseteq A \mid \text{V complete} \}\rvert$ and $j$ is the number of clauses.
By construction $A \weaklytriggers \{n_{i,1}, \ldots, n_{i,v} \}$, for each $1 \leq i \leq j$.
Since for every complete extension this set contains at least one norm (directly) triggered by it, 
it holds that $\wo(\{n_{i,1}, \ldots, n_{i,v} \}) \in \bigcap \{ \cn(N(V)) \mid \text{ $V$ is a complete extension of $A$} \}$.
Since $\bigwedge_{1 \leq i \leq k_V} \head(n^V_i) \vdash \varphi$ for every complete extension $V$ of $A$
it follows that $\{ \wo(\{n_{i,1}, \ldots, n_{i,v} \}) \mid 1 \leq i \leq j  \} \vdash \varphi$ as desired.

For the converse direction, let $\mathbf{B}_2(N, A) \vdash \varphi$. 
By definition, for some sets of norms $N^{(1)}, \ldots, N^{(k)} \subseteq N$ such that $A \weaklytriggers N^{(i)}$, it holds that 
$\bigwedge_{1 \leq i \leq k} \wo\big(N^{(i)}\big) \vdash \varphi$.
As each $N^{(i)}$ is weakly triggered by A, it holds that $V \vdash \bigvee_{n \in N^{(i)}}  \body(n)$ for every complete
extension $V \supseteq A$ of $A$.
Since the extensions $V$ are complete, for each $N^{(i)}$ there exists some $n_i \in N^{(i)}$ such that $V \vdash \body(n_i)$.
Thus, for a fixed complete extension $V$, it holds that $\head(n_i) \in N(V)$ and, since $\{\head(n_i)\} \vdash \wo(N^{(i)})$, also 
$\cn(N(V)) \vdash \wo(N^{(i)})$. As every weak output $\wo(N^{(i)})$ is output by every complete $V$, it follows
that $\out_2(N,A) \vdash \varphi$.
\end{proof}

Alternatively, the task of constructing the basis $\mathbf{B}_2(N, A)$ of $\out_2(N,A)$ can be reduced to finding the \emph{minimal}
weakly triggered sets. A set $N^\prime \subseteq N$, $N^\prime \neq \emptyset$,
is called \emph{minimally weakly triggered by $A$}, written
$A \minweaklytriggers N^\prime$, iff $A \weaklytriggers N^\prime$
and for any proper non-empty subset $M \subset N^\prime$ it is not the case that $A \weaklytriggers M$.

In the above construction, we can restrict ourselves to minimally weakly triggered sets without loss of generality
as they entail each output of every non-minimally weakly triggered superset:
\begin{theorem}\label{theorem:out2b}
$\mathbf{B}_2(N, A) = \big\{ \mathrm{wo}(N^\prime) \mid A \minweaklytriggers N^\prime \text{ for some } N^\prime \subseteq N \big\}$  
\qed
\end{theorem}

For basic reusable output $\out_4$, as shown by Makinson and van der Torre, the output set can be characterized using $\out_2$ and so-called \emph{materialization}~\cite{DBLP:journals/jphil/MakinsonT00}:

\begin{lemma}
\label{lemma:reduction4to2}
For each normative system $N$ and input $A$ it holds
that $\out_4(N,A) = \out_2(N, A \cup m(N))$, where $m(N)$ is the
\emph{materialization of $N$} given by $m(N) := \{\neg a \lor x \mid (a,x) \in N \}$.
\qed
\end{lemma}

Following this alternative representation, the basis $\mathbf{B}_4(N, A)$ for $\out_4$ is defined as:
\begin{equation*}
\mathbf{B}_4(N, A) := \mathbf{B}_2(N, A \cup m(N))
\end{equation*}

\begin{corollary}\label{theorem:out4}
$\phi \in \out_4(N,A)$ if and only if $\mathbf{B}_4(N, A) \vdash \phi$. \qed
\end{corollary} 

\paragraph*{Construction of bases.}
The computational construction of $\mathbf{B}_2(N, A)$ is more involved than for simple-minded and reusable
output. To that end, weakly triggered sets are related with the well-known notion of minimal unsatisfiable subsets:
\begin{definition}[Minimal unsatisfiable subset]
  Let $C$ be an unsatisfiable set of formulas. A subset $C^\prime \subseteq C$ is a minimal unsatisfiable subset (MUS)
  of $C$ if and only if $C^\prime$ is unsatisfiable and any proper subset of $C^\prime$ is satisfiable.
\end{definition}

MUSes play an important role in many contexts such as model checking~\citep{DBLP:conf/fmcad/GhassabaniWG17},
requirement analysis~\citep{DBLP:journals/fac/BarnatBBBBK16,DBLP:conf/issta/Bendik17} and
software and hardware assessment~\citep{DBLP:journals/tsmc/HanL99,DBLP:conf/haskell/StuckeySW03}.
Whereas SAT solvers usually compute only whether a given set of formulas is satisfiable or unsatisfiable
(including certificates that substantiate this claim), 
MUS enumeration tools extend this functionality by providing means to extract the concrete MUSes that 
cause the initial set of formulas to be unsatisfiable.
In the given context, MUSes provide a specific minimal choice of norms that are weakly triggered by $A$.
Let $\overline{\{\varphi_1, \ldots, \varphi_k\}} = \{\neg \varphi_1, \ldots, \neg \varphi_k\}$
denote the complementation (element-wise negation) of a set of formulas.

\begin{lemma} \label{lemma:mus}
Let $N^\prime \subseteq N$ be a subset of norms, $A$ some input, and 
let $D = A \cup \overline{\bodies(N)}$.
It holds that
\begin{enumerate}
  \item if $A \minweaklytriggers N^\prime$ then $A^\prime \cup \overline{\bodies(N^\prime)}$ is a MUS of $D$ for some $A^\prime \subseteq A$, and
  \item if, for some $A^\prime \subseteq A$, $A^\prime \cup \overline{\bodies(N^\prime)}$ is a MUS of $D$, then $A \weaklytriggers N^\prime$.
\end{enumerate}
\end{lemma}
\begin{proof} 
For (1): If $A \minweaklytriggers N^\prime$ then 
$\big( \bigwedge A  \supset \bigvee_{n \in N^\prime}  \body(n) \big)$   
is a tautology. 
It follows that $A \cup \big\{  \bigwedge_{n \in N^\prime} \neg \body(n) \big\}$ is unsatisfiable,
and hence $A \cup \overline{\bodies(N^\prime)}$ as well. Since $N^\prime$ is minimally weakly triggered,
$A^\prime \cup \overline{\bodies(N^\prime)}$ is a MUS for some $A^\prime \subseteq A$.

For (2): If $A^\prime \cup \overline{\bodies(N^\prime)}$ is a MUS of $D$ it directly follows
that $A^\prime \weaklytriggers N^\prime$. Since $\vdash$ is monotone, $A \weaklytriggers N^\prime$
holds as well.
\end{proof}

Note that not every MUS corresponds to a minimally weakly triggered set of norms, but only a (non-minimally)
weakly triggered one.\footnote{Thanks to an (anonymous) reviewer who pointed this out.}
A simple example is as follows: Let $A = \{a,a\lor b\}$ and $N = \{(a,x),(b,y)\}$. Then, the set $D$
given by
\begin{equation*}
  D = \{a, a \lor b, \neg a, \neg b \}
\end{equation*}
is unsatisfiable, and has the following two MUSes:
\begin{equation*}
  M_1 = \{a,\neg a \}, \qquad M_2 = \{a \lor b, \neg a, \neg b\}.
\end{equation*}
Whereas both $A \weaklytriggers \{(a,x)\}$ (using input $a$ as witness) and $A \weaklytriggers \{(a,x),(b,x)\}$ (using
input $a \lor b$ as witness),
only the first MUS corresponds to a minimally weakly triggered set, i.e., $A \minweaklytriggers \{(a,x)\}$.
If, however, not all MUSes are considered, but only those MUSes $A^\prime \cup \overline{\bodies(N^\prime)}$ with $\subseteq$-minimal 
$\overline{\bodies(N^\prime)}$, the correspondence can be established.
For the results of this work, this asymmetry does not matter.

Well-developed
MUS enumeration and extraction approaches~\citep{belov2012muser2,bacchus2016finding,bendik2018recursive,narodytska2018core,bendik2020must}
can be used to construct $\mathbf{B}_2(N, A)$ by enumerating all the MUSes, and to reconstruct the weakly
triggered sets of norms.
The procedure $\texttt{weaklyTriggeredNorms}(A,N)$ implements this as follows:
\begin{equation*}
\texttt{weaklyTriggeredNorms}(A, N) := \{ N^\prime \subseteq N \mid A^\prime \cup \overline{\bodies(N^\prime)} \textit{ is a MUS of } D \},
\end{equation*}
where $D$ is defined as in Lemma~\ref{lemma:mus} above, and $A^\prime \subseteq A$ is some subset of the input.
Finally we set $\mathbf{B}_2(N, A) := \{ \wo(N^\prime) \mid N^\prime \in \texttt{weaklyTriggeredNorms}(A, N) \}$.
In the worst case there are exponentially many MUSes, so enumerating all of them can take exponentially many
SAT calls.

\subsection{Reductions for throughput operators $\out_i^+$}

The bases $\textbf{B}_i^+$ of the output sets produced by the throughput variants $\out_i^+$ can
be constructed using the results from above and identities from the literature.

The throughput operators $\out_1^+$ and $\out_3^+$ based on simple-minded and reusable output, respectively,
can be reduced to their non-throughput variants.
Similarly, the basic output $\out_2^+$ and basic reusable output $\out_4^+$ operators
collapse into classical logic:
\begin{lemma}[\citep{DBLP:journals/jphil/MakinsonT01}]
  Let $N$ be a set of conditional norms, and $A \subseteq L$.
  \begin{itemize}
    \item For $i \in \{1,3\}$ it holds that $\out_i^+(N,A) = \cn(A \cup \out_i(N,A))$.
    \item For $i \in \{2,4\}$ it holds that $\out_i^+(N,A) = \cn(A \cup m(N))$. \qed
  \end{itemize}
\end{lemma}

\paragraph*{Construction of bases.}
Let $\mathbf{B}_i^+(N,A)$ denote the base of $\out_i^+$, for $i \in \{1,2,3,4\}$. It is
given by:
\begin{equation*}\begin{split}
  \mathbf{B}_1^+(N,A) &:= A \cup \mathbf{B}_1(N,A)  \\
  \mathbf{B}_2^+(N,A) &= \mathbf{B}_4^+(N,A) := A \cup m(N) \\
  \mathbf{B}_3^+(N,A) &:= A \cup \mathbf{B}_3(N,A) \\
\end{split}\end{equation*}

\begin{theorem}
$\phi \in \out^+_i(N,A)$ if and only if $\mathbf{B}^+_i(N, A) \vdash \phi$, for
$i \in \{1,2,3,4\}$. \qed
\end{theorem}

In the cases of $i \in \{1,3\}$ the bases $\mathbf{B}_i^+(N,A)$ are hence given by augmenting
the non-throughput bases with the input; for $i \in \{2,4\}$ they are simply given by the classical
materialization.


\newcommand{\prefout}{\mathit{prefout}}
\subsection{Reduction for constrained output}
The reduction of constrained I/O logic reasoning to SAT builds upon the procedures
for generating a finite basis for the underlying (unconstrained) output set.
Recall that in constrained
I/O logic outputs are generated by aggregating over a family of output sets that
are consistent with the given constraints. The selection of the appropriate
sets of norms can, again, be determined by SAT.

For a given normative system $N$, input $A$ and set of constraints $C$, 
the generation procedure of $\mathit{maxfamily}_{\out_i}(N,A,C)$ with respect to an output
operator $\out_i$, $i \in \{1,2,3,4\}$, is displayed in Fig.~\ref{fig:constrained}.
The procedure yields ${\subseteq}$-maximal subsets $M \subseteq N$ of $N$
such that $\mathbf{B}_i(M, A) \cup C$ satisfiable.
By Theorems~\ref{theorem:out1}, \ref{theorem:out3}, \ref{theorem:out2} and \ref{theorem:out4}
this is equivalent to asserting consistency with respect to the output generated by
$\out_i$.
In Fig.~\ref{fig:constrained} the set $U$ collects unprocessed sets of norms
that still have to be checked for consistent output, while $P$ collects norm sets for which
consistent output has already been asserted. The variable $T$ used in the body of the loop
acts as temporary storage. If consistency of $M \subseteq N$ 
(with respect to $C$) can be established, $M$ is added to $P$ and contained in the $\mathit{maxfamily}$.
If, in contrast, the output of $M$ under $A$ is inconsistent with $C$, all
$|M|-1$-subsets of $M$ (i.e., subsets that have one element less) are inserted into $U$
for subsequent assessment.
The procedure terminates as $N$ is finite and only has finitely many subsets, and
the subsets added to $U$ are monotonously decreasing in cardinality.
$\mathit{outfamily}(N,A,C)$ is represented by a collection of finite bases, i.e., 
by $\big\{ \mathbf{B}_i(N^\prime, A) \mid N^\prime \in \mathit{maxfamily}_\out(N,A,C)\big\}$.
Finite bases of $\out^C_\cup(N,A)$ and $\out^C_\cap(N,A)$ are generated by 
computing the join and meet of the individual bases, respectively.

\begin{figure}[tb]
\centering
\begin{lstlisting}[frame=lines,basicstyle=\small\ttfamily,mathescape,morekeywords={if,for,endfor,endif,then}]
$U$ := $\{ N \}$  
$P$ := $\emptyset$
while $U \neq \emptyset$ do
  $T$ := $\left\{ M \in U \mid \sat\big(\mathbf{B}(M,A) \cup C\big) \right\}$
  $P$ := $P \cup T$
  $U$ := $\left\{M^\prime \subseteq M \mid M \in (U \setminus T) \textrm{ and } \lvert M^\prime \rvert = \lvert M \rvert -1 \right\}$ 
end while
return $P$
\end{lstlisting}
\caption{Procedure for generating $\mathit{maxfamily}_\out(N,A,C)$ where $\mathbf{B}(N,A)$
is the finite basis for the output operator $\out$.\label{fig:constrained}}
\end{figure}

In the worst case, every subset of $N$ needs to be checked individually for
consistent output in Fig.~\ref{fig:constrained}, i.e., when every singleton subset of $N$
yields inconsistent output. This
requires $\mathcal{O}\big(2^{ \lvert N \rvert }\big)$ SAT queries in the worst case; reflecting the 
very cost-intensive nature of constrained I/O logic reasoning, being
$\text{NP}^{\text{NP}}$-complete resp.\ $\text{coNP}^{\text{NP}}$-complete~\citep{DBLP:journals/japll/SunR17}.

The procedure in Fig.~\ref{fig:constrained} is mainly for illustrative purposes, and it
can be optimized to utilize more sophisticated techniques such as MaxSAT-based enumeration of \emph{minimal correction sets}~\citep{DBLP:conf/sat/LiffitonS09,DBLP:conf/hvc/MorgadoLM12},
i.e., sets of formulas whose removal will make the set of formulas satisfiable.
Of course, the theoretical worst-case complexity remains the same in any case;
still this would yield a more effective procedure in practice.
The technical details are omitted here for brevity.

\subsection{Reductions for other extensions}

Recall that, for some output operator $\out$,
$\mathit{maxfamily}_{\out}(N,A,C)$ collects all the subsets of norms of $N$ whose
output under $A$ is consistent with $C$, and that $\mathit{outfamily}_{\out}(N,A,C)$ collects the
respective sets of outputs. An output is then produced by aggregating those  sets into a single output set (called the \textit{net output} by Makinson and van der Torre), e.g., by credulous output (union) or skeptical output (intersection). A natural generalization of producing the output is then the following:
Let $\gout$ be the \emph{generalized net output operator} defined as follows:
\begin{equation*}
  \gout(N,A,C,\out,\mathcal{S}, \mathcal{A}) := \mathcal{A}\Big(\big\{ \out(M,A) \mid M \in \mathcal{S}\big(\mathit{maxfamily}_\out(N,A,C)\big) \big\}\Big)
\end{equation*}
where $N$ is a set of norms, $A$ is an input, $C$ is a set of constraints, and $\out$ some output operator.
Furthermore, $\mathcal{S}: 2^{2^{N}} \to 2^{2^N}$ is a selection function on sets of norm sets, and 
$\mathcal{A}: 2^{2^L} \to 2^L$ an aggregation function.

Credulous and skeptical output are simply simulated by
letting $\mathcal{A}$ be set union resp. set intersection, and $\mathcal{S}$ the identity function. 

There exist a notable generalization of constrained I/O logics 
where preference orderings ${\succeq} \subseteq N \times N$ (transitive-reflexive relations)
between conditional norms are employed for generating a family of \emph{preferred} (wrt.\ $\succeq$)
maximal subsets of $N$ and respective output families~\citep{DBLP:journals/ail/Parent11}.

For some set of constraints $C$, some lifting $\succeq^*$ of $\succeq$ to sets of norms,
output operator $\out$ and aggregation function $\star \in \{\cup,\cap\}$,
the definition of preferred output $\out^{C,\succeq^*}_\star$ is given by:
\begin{equation*}
\out^{C,\succeq^*}_\star(N,A) := \star \big\{ \out(M,A) \mid M \in \mathit{maxfamily}_\out(N,A,C), M \text{ maximal wrt. } \succeq^* \big\}
\end{equation*} 

$\out^{C,\succeq^*}_\star(N,A)$ can be reformulated in terms of $\gout$ as follows:
\begin{equation*}
\out^{C,\succeq^*}_\star(N,A) = \gout(N, A, C, \out, \textit{maximals}(\succeq^*), \star)
\end{equation*} 
where $\textit{maximals}(\succeq^*)$ is the function uniquely determined by $\succeq^*$ that gives the $\succeq^*$-maximal subsets of a given set of norm sets.

Further selection functions can be studied in this context, e.g., selecting largest or smallest members of the \textit{maxfamily}, or only those members that do (or do not) contain some specific norm, etc. Likewise, further aggregation functions can be studied wrt. $\gout$.

Computing $\gout$ can integrated in the overall automation framework in a straight-forward way, and hence extensions such as preferred output as defined above. The pseudo-code for this procedure is displayed in Fig.~\ref{fig:gout}.\footnote{Note that it is assumed that both $\mathcal{S}$ and $\mathcal{A}$ are computable.} Essentially, calculating $\gout$ first calculates the base of the respective \textit{maxfamily}, then only those members as selected by $\mathcal{S}$ are retained, and then the computation proceeds as usual by computing the respective finite bases of the selected sets, returning the net output as defined by $\mathcal{A}$.

\begin{figure}[tb]
\centering
\begin{lstlisting}[frame=lines,basicstyle=\small\ttfamily,mathescape,morekeywords={if,for,endfor,endif,then}]
maxfamily := $\mathit{maxfamily}_\out(N,A,C)$ 
selected  := $\mathcal{S}($maxfamily$)$ 
outs      := $\big\{ \mathbf{B}(M,A) \mid M \in \texttt{selected} \big\}$
return $\mathcal{A}(\texttt{outs})$
\end{lstlisting}
\caption{Procedure  for generating $\gout(N,A,C,\out,\mathcal{S},\mathcal{A})$.
In the first line, the procedure for generating the \textit{maxfamily} from Fig.\ref{fig:constrained} is reused.\label{fig:gout}}
\end{figure}

The computation of $\gout$ has been implemented in \rio, as illustrated for the case of preferred output in the case study in Sec.~\ref{sec:casestudy}.

\section{TPTP Representation}
\rio\ is aligned with the TPTP standards. 
The TPTP \emph{(Thousands of Problems for Theorem Proving)} library and infrastructure~\citep{Sut17} is the core
platform for the development and evaluation of contemporary automated theorem proving (ATP) systems. It provides (i) a comprehensive collection
of benchmark problems for ATP systems; (ii) a set of tools for problem and solution inspection, pre- and post-processing,
and verification; and (iii) a comprehensive syntax standard for ATP system input and output, including
result reporting. The connection to the TPTP infrastructure is chosen as TPTP is the 
predominant standard for classical ATP systems, and its comprehensive syntax standards
already support future (planned) extesions of \rio\ to quantified logics. Even more, this allows for
the integration of I/O logic reasoning with other (classical and non-classical) 
TPTP-compliant reasoning tools.
The brief TPTP introduction below
is adapted from earlier work~\citep{DBLP:conf/paar/SteenS24}.

\subsection{TPTP Preliminaries.}
The TPTP specifies different ATP system languages varying in their expressivity. In this work we focus on the
\emph{typed first-order form} (TFF), a syntax standard for many-sorted first-order logic~\citep{DBLP:conf/lpar/SutcliffeSCB12},
All input formats are in ASCII plain text,
human-readable and follow Prolog language conventions for simple parsing. An ATP problem generally consists of
type declarations, contextual definitions and premises of the reasoning
task (usually referred to as \emph{axioms}), and a conjecture that is to be proved
or refuted in the given context. The core building block of the ATP problem files
in TPTP languages are
so-called \emph{annotated formulas} of form \ldots \smallskip

\mbox{} \quad {\em language}{\tt (}{\em name}{\tt ,}{\em role}{\tt ,}{\em formula}[{\tt ,}{\em source}[{\tt ,}{\em annotations}]]{\tt ).}
\smallskip

Here, {\em language} is a three-letter identifier for the language format, here {\tt tff}.
The {\em name} is a unique identifier for the annotated formula
but has no other effect on the interpretation of it. The {\em role} field specifies whether the {\em formula} should be interpreted
as an assumption (role {\tt axiom}), a type declaration (role {\tt type}), a definition (role {\tt definition}) or as
a formula to be proved (role {\tt conjecture}). The {\em formula} is an ASCII representation of the respective logical
expression, where predicate and function symbols are denoted by strings that begin with a lower-case letter,
variables are denoted by strings starting with an upper-case letter, the logical connectives 
$\neg$, $\land$, $\lor$, $\rightarrow$, $\leftrightarrow$ are represented by {\tt {\char`\~}}, {\tt \&}, {\tt \verb!|!}, {\tt =>} and {\tt <=>}, respectively.
Quantifiers $\forall$ and $\exists$ are expressed by {\tt !} and {\tt ?}, respectively, followed by a list of
variables bound by it. The TPTP defines several interpreted symbols starting with a {\tt \$}-sign, 
including logical constants such as {\tt \$true} and {\tt \$false} for truth and falsehood, respectively. 
The type {\tt \$i} is the type of individuals and {\tt \$o} is the type of Booleans.
Explicit types of symbols may be dropped. In this case, default types of symbols are assumed
as specified by the TPTP standard as follows: If a symbol {\tt c} occurs at term position with $n$ applied arguments
the default type of {\tt c} is {\tt (\$i * \ldots\ * \$i) > \$i}, representing the type of a $n$-ary
function symbol. If a symbol {\tt c} occurs at formula position with $n$ applied arguments
the default type of {\tt c} is {\tt (\$i * \ldots\ * \$i) > \$o}, representing the type of a $n$-ary
predicate symbol. Arguments are applied using parentheses as in
\verb|f(a,b)|, where \verb|f| is some function symbol, and \verb|a| and \verb|b| are some terms (of appropriate type).
Finally, the {\em source} and {\em annotations} are optional and uninterpreted extra-logical information that can be assigned to the
annotated formula, e.g., about its origin, its relevance, or other properties.
An example formula in TFF format is as follows:
\begin{verbatim}
  tff(union_def, axiom, ! [S, T, X]: (
                            member(X, union(S,T)) <=>
                           ( member(X, S) | member(X, T) ) ),
                        source('definitions.ax'),
                        [relevance(1.0)]).
\end{verbatim}
In this example, a TFF annotated formula of name {\tt union\_def} is given that describes an axiom 
giving a fundamental property of set  union and some auxiliary information about it.
A complete description of the TPTP input languages, including the syntax BNF,
is provided at the TPTP web page\footnote{\url{http://tptp.org}}. For further information about available tools,
problems sets and an extensive description of the different  languages, we refer to the
literature~\citep{Sut17}.

TPTP traditionally focused on classical logic.
Only recently, TPTP formats were extended towards non-classical logics as well~\citep{DBLP:conf/paar/SteenFGSB22}.
For this purpose, the TFF language has been extended with 
expressions
of the form \ldots \\
\mbox{} \quad {\tt \{}{\em connective\_name}{\tt \} @ }{\tt (}{\em arg$_1$}{\tt ,}\ldots{\tt,}{\em arg$_n$}{\tt )},\\
where {\em connective\_name} is either a TPTP-defined name (starting with a {\tt \$} sign)
or a user-defined name (starting with two {\tt \$} signs) for a non-classical operator, and
the {\em arg$_i$} are terms or formulas to which the operator is applied. TPTP-defined connectives
have a fixed meaning and are documented by the TPTP; the interpretation of user-defined connectives is provided by third-party systems, environments, or documentation. The so enriched TPTP language is denoted NXF (non-classical extended first-order form).

As second new component, non-classical NXF adds \emph{logic specifications}
to the language. They are annotated formulas of form \ldots \\
\mbox{} \quad {\tt tff(}{\em name}{\tt , logic,} {\em logic\_name }{\tt == [}{\em options}{\tt ]} {\tt ).} \\
where {\tt logic} is the new TPTP role, {\em logic\_name} is a TPTP-defined or user-defined designator for a logical
language and {\em options} are comma-separated key-value pairs that fix the
specific logic based on that language.
An in-depth overview of the new format is presented in~\citep{DBLP:conf/paar/SteenFGSB22,DBLP:conf/paar/SteenS24}.

\subsection{A Format for I/O logic reasoning.}
The problem input format of \rio\ is TPTP NXF as introduced above.
Currently, only the propositional fragment of NXF
is supported by \rio\ as the underlying formalism for I/O logic reasoning. 

An I/O logic problem file starts with a semantics specification, i.e., a NXF formula
of role \emph{logic}, as follows:
\begin{lstlisting}[basicstyle=\ttfamily, mathescape]
tff($\langle \textit{name} \rangle$, logic, (
    $\mbox{\textdollar}\mbox{\textdollar}$iol == [ $\mbox{\textdollar}\mbox{\textdollar}$operator == $\langle \textit{output operator} \rangle$,
              $\mbox{\textdollar}\mbox{\textdollar}$throughput == $\langle \mbox{\textdollar}\text{true} \textit{ or } \mbox{\textdollar}\text{false} \rangle$,
              $\mbox{\textdollar}\mbox{\textdollar}$constrained == $\langle \textit{output aggregate function} \rangle$,
              $\mbox{\textdollar}\mbox{\textdollar}$constraints == [$\langle \textit{formulas} \rangle$],
              $\mbox{\textdollar}\mbox{\textdollar}$preference == [$\langle \textit{tuples of names} \rangle$] ] )).
\end{lstlisting}
In this specification, the $\langle ... \rangle$ are placeholders for the possible arguments:
\begin{itemize} 
  \item \lstinline[basicstyle=\ttfamily, mathescape]|$\langle \textit{output operator} \rangle$|
        $\in \{$ \lstinline[basicstyle=\ttfamily, mathescape]|$\mbox{\textdollar}\mbox{\textdollar}$out1|,
                 \lstinline[basicstyle=\ttfamily, mathescape]|$\mbox{\textdollar}\mbox{\textdollar}$out2|,
                 \lstinline[basicstyle=\ttfamily, mathescape]|$\mbox{\textdollar}\mbox{\textdollar}$out3|,
                 \lstinline[basicstyle=\ttfamily, mathescape]|$\mbox{\textdollar}\mbox{\textdollar}$out4| $\}$,
  \item \lstinline[basicstyle=\ttfamily, mathescape]|$\langle \textit{output aggregate function} \rangle$|
        $\in \{$ \lstinline[basicstyle=\ttfamily, mathescape]|$\mbox{\textdollar}\mbox{\textdollar}$credulous|,
                 \lstinline[basicstyle=\ttfamily, mathescape]|$\mbox{\textdollar}\mbox{\textdollar}$skeptical| $\}$,
  \item \lstinline[basicstyle=\ttfamily, mathescape]|$\langle \textit{formulas} \rangle$|
        is a comma-separated sequence of formulas, and
  \item \lstinline[basicstyle=\ttfamily, mathescape]|$\langle \textit{tuples of names} \rangle$| is a comma-separated sequence
  of tuples that contain names of annotated formulas.
\end{itemize}
The entries 
\lstinline[basicstyle=\ttfamily, mathescape]|$\mbox{\textdollar}\mbox{\textdollar}$throughput|,
\lstinline[basicstyle=\ttfamily, mathescape]|$\mbox{\textdollar}\mbox{\textdollar}$constrained|
\lstinline[basicstyle=\ttfamily, mathescape]|$\mbox{\textdollar}\mbox{\textdollar}$constraints|, and
\lstinline[basicstyle=\ttfamily, mathescape]|$\mbox{\textdollar}\mbox{\textdollar}$preference|
are optional. 
\lstinline[basicstyle=\ttfamily, mathescape]|$\mbox{\textdollar}\mbox{\textdollar}$throughput|
defaults to \lstinline[basicstyle=\ttfamily, mathescape]|$\mbox{\textdollar}$false| (when given no value)
so that the output operators are used without throughput when this parameter is omitted.
Giving no value for \lstinline[basicstyle=\ttfamily, mathescape]|$\mbox{\textdollar}\mbox{\textdollar}$constrained|
means that unconstrained I/O logic reasoning is used.
When using constrained output, the \lstinline[basicstyle=\ttfamily, mathescape]|$\mbox{\textdollar}\mbox{\textdollar}$preference|
parameter specifies the preference relation $\succeq$ between the given norms. Here, all norms within a single tuple are equally
preferred (wrt.\ $\succeq$), and preferred over every norm in any of the following tuples, etc. As an example, an entry of the form
\lstinline[basicstyle=\ttfamily, mathescape]|[[name1,name2],[name3],[name4]]|
specifies a preference relation in which the norm with name
\lstinline[basicstyle=\ttfamily, mathescape]|name1|
is equally preferred as norm \lstinline[basicstyle=\ttfamily, mathescape]|name2|, and both are preferred over
\lstinline[basicstyle=\ttfamily, mathescape]|name3| and \lstinline[basicstyle=\ttfamily, mathescape]|name4|, where in turn
\lstinline[basicstyle=\ttfamily, mathescape]|name3| is preferred over \lstinline[basicstyle=\ttfamily, mathescape]|name4|.
Giving a value for  
\lstinline[basicstyle=\ttfamily, mathescape]|$\mbox{\textdollar}\mbox{\textdollar}$constraints|
has no effect if 
\lstinline[basicstyle=\ttfamily, mathescape]|$\mbox{\textdollar}\mbox{\textdollar}$constrained|
is not set. Equivalently, giving a value for
\lstinline[basicstyle=\ttfamily, mathescape]|$\mbox{\textdollar}\mbox{\textdollar}$preference| has no effect
if  \lstinline[basicstyle=\ttfamily, mathescape]|$\mbox{\textdollar}\mbox{\textdollar}$constrained| is not set.
For convenience, the \emph{pseudo formula}
\lstinline[basicstyle=\ttfamily, mathescape]|$\mbox{\textdollar}\mbox{\textdollar}$input|
can be used as value for 
\lstinline[basicstyle=\ttfamily, mathescape]|$\mbox{\textdollar}\mbox{\textdollar}$constraints|
for directly referring to the input $A$.

The remaining problem file is structured as follows: 
Norms are non-classical formulas of role \emph{axiom} of form \ldots
\begin{lstlisting}[basicstyle=\ttfamily, mathescape]
tff($\langle \textit{name} \rangle$, axiom, {$\mbox{\textdollar}\mbox{\textdollar}$norm} @ ($\langle \textit{body} \rangle$, $\langle \textit{head} \rangle$)).
\end{lstlisting}
and inputs are classical TFF formulas of role \emph{hypothesis}.
Additionally, a problem may contain
an arbitrary number of \emph{conjecture}s that express conjectured outputs.

\section{Implementation \label{sec:impl}} 
The prototypical automated reasoning system
\emph{rio} (for \emph{\underline{r}easoner for \underline{i}nput/\underline{o}utput logic})
implements decision procedures based on the SAT reductions presented above.
It is written in Scala and uses PicoSAT~\citep{DBLP:journals/jsat/Biere08} as a SAT solver backend and
MUST~\citep{bendik2020must} for MUS enumeration.
\emph{rio} is open-source and available under BSD license at GitHub\footnote{\url{https://github.com/aureleeNet/rio}.}
and via Zenodo~\citep{rio}.

\subsection{Application Examples}

\begin{figure}[tb]
  \centering
\begin{lstlisting}[frame=lines,basicstyle=\small\ttfamily,morekeywords={if,for,endfor,endif,then},escapechar=\%]
tff(simple, logic, 
    $$iol == [ $$operator == $$out4 ] ).
    
tff(norm1, axiom, {$$norm} @ (parking, ticket %$\mathtt{\vert}$% fine) ).
tff(norm2, axiom, {$$norm} @ (ticket, pay) ).
tff(norm3, axiom, {$$norm} @ (fine, pay) ).

tff(input1, hypothesis, parking).
\end{lstlisting}
\caption{Problem representation in unconstrained I/O logic without conjectures. The finite basis of
$\out_4(N,A)$, with $A = \{\textit{parking}\}$ and $N = \{ (\textit{parking}, \textit{ticket} \lor \textit{fine}), (\textit{ticket}, \textit{pay}), (\textit{fine}, \textit{pay}) \}$, is calculated. The output is depicted in Fig.~\ref{fig:result}.}
\label{fig:example}
\end{figure}

 \begin{figure}[tb]
     \centering
\begin{lstlisting}[frame=lines,basicstyle=\small\ttfamily,morekeywords={if,for,endfor,endif,then}]
tff(my_spec, logic, 
    $$iol == [ $$operator == $$out3,
               $$constrained == $$skeptical,
               $$constraints == [~ helping] ] ).
              
tff(norm1, axiom, {$$norm} @ ($true, helping) ).
tff(norm2, axiom, {$$norm} @ (helping, telling) ).
tff(norm3, axiom, {$$norm} @ (~helping, ~telling) ).

tff(fact_not_helping, hypothesis, ~helping).

tff(conjecturedOutput1, conjecture, ~telling).
tff(conjecturedOutput1, conjecture, ~helping).
\end{lstlisting}
\caption{Problem representation in constrained I/O logic with two conjectures. The problem is to decide whether
     $\neg\textit{telling} \in \out^C_\cap(N,A)$ and $\neg\textit{helping} \in \out^C_\cap(N,A)$, where $\out = \out_3$,
     $A = \{ \neg\textit{helping} \}$ and $N = \{ (\top, \textit{helping}), (\textit{helping}, \textit{telling}), (\neg\textit{helping}, \neg\textit{telling}) \}$. The output is depicted in Fig.~\ref{fig:result2}. }
\label{fig:example2}
\end{figure}

Two example I/O logic problems are given in Figures~\ref{fig:example} and \ref{fig:example2}.
In the first example, cf.\ Fig.~\ref{fig:example}, three norms are encoded:
$(\mathit{parking}, \mathit{ticket} \lor \mathit{fine})$,
$(\mathit{ticket}, \mathit{pay})$,
$(\mathit{fine}, \mathit{pay})$ that represent the expressions
\emph{if I park my car, I ought to get a parking ticket or get a fine},
\emph{if I get a parking ticket, I ought to pay it}, and
\emph{if I get a fine, I ought to pay it}, respectively.
The problem furthermore specifies the input $\mathit{parking}$.

The second example, cf.\ Fig.~\ref{fig:example2}, encodes Chisholm
paradox as introduced in Example~\ref{example:chisholm}. In the problem statement, two outputs
are conjectured: $\neg\mathit{telling}$ and $\neg\mathit{helping}$.

\rio\ returns results according to the SZS standard~\citep{Sut17}.
There are two main use cases: 
If no conjecture is given, \rio\ will return SZS status
\lstinline[basicstyle=\ttfamily, mathescape]|Success|, generate a finite base of all outputs
and return it as SZS output \lstinline[basicstyle=\ttfamily, mathescape]|ListOfFormulae|.
If the problem contains at least one conjecture, \rio\ will return one SZS status
per conjecture as follows: SZS status 
\lstinline[basicstyle=\ttfamily, mathescape]|Theorem| for conjectured outputs that
are indeed in the output set, and SZS status 
\lstinline[basicstyle=\ttfamily, mathescape]|CounterSatisfiable|
for conjectured outputs that are, in fact, not contained in the output set.

The output given by \rio\ for each of the two examples from Figures~\ref{fig:example} and~\ref{fig:example2}
is displayed at Figures~\ref{fig:result} and~\ref{fig:result2}:
In Fig.~\ref{fig:result}, the elements of the finite basis of the output are given
line-by-line, representing $\mathbf{B}_4(N, A) = \{ \mathit{pay}, \mathit{ticket} \lor \mathit{fine} \}$.
In Fig.~\ref{fig:result2}, each of the conjectures is assessed
by a dedicated SZS status line, indicating that $\neg \mathit{telling} \in \out_\cap^C(N,A)$
and $\neg \mathit{helping} \notin \out_\cap^C(N,A)$.

\begin{figure}[tb]
  \centering
\begin{lstlisting}[frame=lines,basicstyle=\small\ttfamily,morekeywords={if,for,endfor,endif,then},escapechar=?]
% SZS status Success for parking.p
% SZS output start ListOfFormulae for parking.p
ticket ?$\mathtt{\vert}$? fine
pay
% SZS output end ListOfFormulae for parking.p
\end{lstlisting}
\caption{Result for the unconstrained I/O logic problem from Fig.~\ref{fig:example}. The finite basis
of the output set is listed as SZS output, one element per line.}
\label{fig:result}
 \end{figure}

 \begin{figure}[tb]
     \centering
\begin{lstlisting}[frame=lines,basicstyle=\small\ttfamily,morekeywords={if,for,endfor,endif,then},escapechar=?]
% SZS status Theorem for chisholm.p: conjecturedOutput1
% SZS status CounterSatisfiable 
                     for chisholm.p: conjecturedOutput2
\end{lstlisting}
\caption{Result for the constrained I/O logic problem from Fig.~\ref{fig:example2}. For
every conjecture a respective SZS status is printed.}
\label{fig:result2}
\end{figure}

\section{rio: Case Study \label{sec:casestudy}}
The implemented reasoning methods are now applied to a motivating case study originating from a ruling
of the German federal constitutional court (\emph{Bundesverfassungsgericht}, BVerfG) limiting the jurisdictional competence
of the Court of Justice of the European Union (ECJ) in the context of a bond buying programme of the European Central Bank (ECB).\footnote{
Decision of May 5, 2020, 2 BvR 859/15. See \url{http://www.bverfg.de/e/rs20200505_2bvr085915en.html}.}
The ruling illustrates a major difference in the understanding and application of European norms,
in particular the German BVerfG accuses the ECJ of having omitted a key step within its legal process.
However, following the standard reading of legal terminology, the fundamental laws governing the work of
the ECJ prohibit the ECJ from undertake exactly this step. Even more, the German BVerfG has made
it clear in prior rulings that it acknowledges the fundamental principles that underlie this prohibition.
The application of rio for studying this conflict (to some extent) is presented in the following case study.
It is translated and suitably adapted from earlier work~\citep{ASR23}.

A brief summary of the decision's context is as follows: Following a constitutional appeal about the legitimacy of the bond buying programme
of the European Central Bank the BVerfG, within the scope of a so-called preliminary ruling procedure (\emph{Vorabentscheidungsverfahren}), calls upon the ECJ to provide an interpretation of the relevant
treaty law. The ECJ decides that the bond buying programme of the ECB does not conflict with treaty law.
Then, remarkably, the German BVerfG does not follow the ECJ as it contains, according to the BVerfG,
significant methodical errors so that it cannot (and indeed must not) follow its decision. 

In general, the German BVerfG follows German legal methods, whereas it is accepted that
the ECJ works with individual (European) legal methods. While the systems are indeed different, they necessarily have overlapping
concepts such as the differentiation between interpretation of law (an abstract notion) and application of law (a concrete case-oriented notion).
Roughly, law interpretation speaks about how to generally understand the scope of legal norms and their intended meaning. This is
independent from concrete cases, hence considered an abstract-general notion. On the other hand, law applications speaks about, among others,
how to select the correct norms for a concrete case, how to weigh legal goods, and how to produce legal consequences (hence considered a concrete-individual notion). Within the
scope of preliminary ruling procedures, the ECJ is allowed to do the former but not the latter.

A model of the case facts and relevant norms is now constructed in I/O logic such that the decision can be inspected from different perspectives using rio.
In particular, the case study highlights the flexibility of the I/O logic framework for reasoning in
unconstrained, constrained, and preference-based scenarios; and also its uniform representation
in the proposed TPTP-based syntax format.
The use of I/O logic for normative reasoning in multi-agent systems
is advocated, e.g., by Boella and van der Torre~(\citeyear{boella4constitutive}).
For the sake of simplicity, the model here only shows a minimal viable setup as relevant
for the case study.

\subsection{Case facts}

The propositional atoms used as vocabulary for representing the case facts and the relevant norms are displayed and described in Table~\ref{table:casestudy:atoms}. Using these atoms, the case facts are represented as input $A = \{\eqref{fact:a1}, \ldots, \eqref{fact:a7} \}$ as follows:

\begin{table}[tb]
\centering
\begin{tabular}{p{.38\textwidth}|p{.55\textwidth}}
\multicolumn{1}{c|}{\textbf{Propositional atom}} & \multicolumn{1}{c}{\textbf{Short description}} \\
\hline
$\mathit{abstractAssessmentECJ}$ & 
The ECJ applies a abstract-general assessment.
\\ \hline
$\mathit{concreteAssessmentECJ}$ & 
The ECJ applies a concrete-individual assessment.
\\ \hline
$\mathit{interpretationECJ}$ & 
The ECJ interpretes law.
\\ \hline
$\mathit{applicationECJ}$ & 
The ECJ applies law.
\\ \hline
$\mathit{methodicallySoundECJ}$ & 
The work on the ECJ deemed methodically sound/correct.
\\ \hline
$\mathit{ultraViresECJ}$ & 
The ECJ acts ultra vires.
\\ \hline
$\mathit{prelimRulingProcECJ}$ & 
The ECJ conducts a preliminary ruling procedure.
\\ \hline
$\mathit{assessmentOfAdequacyECJ}$ & 
The ECJ conducts an assessment of adequacy.
\\ \hline
$\mathit{assessmentOfProportionalityECJ}$ & 
The ECJ conducts an assessment of proportionality.
\\ \hline
$\mathit{ECBBondBuyingDecision}$ & 
The case at hand is the ECB bond buying decision (subject of the case study).
\\ \hline
$\mathit{BVerfGFollowsECJ}$ & 
The BVerfG follows the ECJ's decision.
\\ \hline
\end{tabular}
\caption{Symbols used for representing the facts and norms of the case study.}
\label{table:casestudy:atoms}
\end{table}

\begin{equation} \tag{A1} \label{fact:a1}
\neg (\textit{abstractAssessmentECJ} \Leftrightarrow \textit{concreteAssessmentECJ})
\end{equation}
A fundamental requirement of the European legal system is that models of law application (subsumption) and models of
law interpretation must differentiate between these two legal concepts.\footnote{
See, e.g., article 267 of the Treaty on the Functioning of the European Union (TFEU) which dictates
that the ECJ may only interpret but not apply law in the context of a preliminary ruling procedure.}
Input \eqref{fact:a1}
maintains this difference by stating that abstract-general assessments must be strictly separated from concrete-individual assessments.

Law interpretation and law application is then defined as an instance of an abstract-general assessment and a concrete-individual
assessment, respectively, by \eqref{fact:a2} and \eqref{fact:a3}.

\begin{align} 
\textit{interpretationECJ} \Rightarrow \textit{abstractAssessmentECJ}   \tag{A2}\label{fact:a2} \\
\textit{applicationECJ} \Rightarrow \textit{concreteAssessmentECJ}   \tag{A3}\label{fact:a3}
\end{align}

A so-called assessment of adequacy is a concrete-individual assessment because it involves law application,
yielding input \eqref{fact:a4}.
\begin{equation} \tag{A4} \label{fact:a4}
\textit{assessmentOfAdequacyECJ} \Rightarrow \neg\textit{abstractAssessmentECJ}
\end{equation}

In the case decision, the BVerfG states: \emph{''This view [...] is no longer tenable from a methodological perspective [...]. [T]he ECJ Judgment itself constitutes an ultra vires act and thus has no binding effect [in Germany].''}\footnote{See Rd. 119 in 2 BvR 859/15, \url{https://www.bverfg.de/e/rs20170718_2bvr085915}.}
An ultra vires act is an act that was executed outside of the scope of one's own competencies (i.e., not having a legal basis).
The conditional that if the ECJ decision is not sound from a methodological perspective, it constitutes an ultra vires act is encoded by \eqref{fact:a5}.
\begin{equation} \tag{A5} \label{fact:a5}
\neg \textit{methodicallySoundECJ} \Rightarrow \textit{ultraViresECJ}
\end{equation}

The relevant ruling of the ECJ that is not followed by the BVerfG is enacted in the scope of a preliminary ruling procedure,
this is formalized by \eqref{fact:a6}.

\begin{equation} \tag{A6} \label{fact:a6}
\textit{prelimRulingProcECJ} \land \textit{ECBBondBuyingDecision}
\end{equation}

For that ruling, an assessment of proportionality was conducted by the ECJ, but an assessment of adequacy was omitted, represented
by \eqref{fact:a7}.
\begin{equation} \tag{A7} \label{fact:a7}
\textit{assessmentOfProportionalityECJ} \land \neg\textit{assessmentOfAdequacyECJ}
\end{equation}

\subsection{Considered Norms}
The relevant norms for the case study are extracted from the relevant legal norms, from implicit knowledge from legal methodology, and from case
decisions. The norm set $N = \{\eqref{norm:n1}, \ldots, \eqref{norm:n9}\}$ is constructed as follows:

If the ECJ's decision is ultra vires, the national court must not follow
it \eqref{norm:n1}; but in general it has to follow the ruling of the ECJ \eqref{norm:n2}.

\begin{align} 
(\textit{ultraViresECJ} &, \neg\textit{BVerfGFollowsECJ})   \tag{N1}\label{norm:n1} \\
(\top &, \textit{BVerfGFollowsECJ}) \tag{N2}\label{norm:n2}
\end{align}

The ECJ may only interpret law in the scope
of a preliminary ruling procedure; but it may not apply law.\footnote{According to article 267 TFEU.}
This is
encoded by norms \eqref{norm:n3} and \eqref{norm:n4}, respectively.
\begin{align} 
(\textit{prelimRulingProcECJ} &, \textit{interpretationECJ})   \tag{N3}\label{norm:n3} \\
(\textit{prelimRulingProcECJ} &, \neg\textit{applicationECJ}) \tag{N4}\label{norm:n4}
\end{align}

It is a usual legal requirement to respect the principle of proportionality\footnote{See also Rd. 123 in 2 BvR 859/15.},
so an assessment of proportionality is required as part of a 
preliminary ruling procedure~\eqref{norm:n5}.
\begin{equation} \tag{N5} \label{norm:n5}
(\textit{prelimRulingProcECJ}, \textit{assessmentOfProportionalityECJ})
\end{equation}

According to the BVerfG, the ECJ's omission of an assessment of adequacy
is a methodological error.\footnote{
See Rd. 119 in 2 BvR 859/15: \emph{''This view [...] is no longer tenable from a methodological perspective given that it completely disregards the actual effects of the PSPP''}.
}
As a consequence, an assessment of adequacy by the ECJ within the scope of a preliminary ruling procedure is, in general, methodologically sound~\eqref{norm:n6}.
\begin{equation} \tag{N6} \label{norm:n6} \small
(\textit{prelimRulingProcECJ} \land \textit{assessmentOfAdequacyECJ}, \textit{methodicallySoundECJ})
\end{equation}
In particular, its omission in the concrete case constitutes a methodological error~\eqref{norm:n7}.
\begin{equation} \tag{N7} \label{norm:n7} 
\small
\begin{split} \small
(\textit{prelimRulingProcECJ} \land \textit{ECBBondBuyingDecision} &\\
\land \neg\textit{assessmentOfAdequacyECJ}&, \neg\textit{methodicallySoundECJ})
\end{split}
\end{equation}

In contrast, the ECJ deems its ruling (without an assessment of adequacy)
to be methodically sound~\eqref{norm:n8}. 

\begin{equation} \tag{N8} \label{norm:n8} 
\small
\begin{split} \small
(\textit{prelimRulingProcECJ} \land \textit{ECBBondBuyingDecision} &\\
\land \neg\textit{assessmentOfAdequacyECJ}&, \textit{methodicallySoundECJ})
\end{split}
\end{equation}

In other cases the ECJ does, however, conduct assessments of adequacy within the scope of preliminary ruling procedures. So this procedure is also methodically sound from the ECJ's perspective~\eqref{norm:n9}.

\begin{equation} \tag{N9} \label{norm:n9} 
\small
\begin{split} \small
(\textit{prelimRulingProcECJ} \land \textit{assessmentOfAdequacyECJ}&, \textit{methodicallySoundECJ})
\end{split}
\end{equation}

Note that norm \eqref{norm:n9} is actually identical to \eqref{norm:n6}, but it
is kept anyway for transparency reasons of the presentation.

\subsection{Assessment using \rio}

It is easy to see that the set $N$ from above contains conflicting norms,
some of which apply in the situation described by $A$.
Indeed \rio\ verifies that $\out_3(N,A) = \cn(\bot)$, i.e. that I/O logic without constraints will not give
meaningful answers:

\begin{Verbatim}[frame=single]
% SZS status Success
% SZS output start ListOfFormulae
  $false
% SZS output end ListOfFormulae
\end{Verbatim}
The input problem for this assessment is displayed in Appendix~\ref{appendix:file1}.

In a constrained setting (without preferences between the individual norms) it holds that
\begin{equation*}\begin{split}
\mathit{maxfamily}_{\out_3}(N,A,C) = \Big\{ &\big\{\eqref{norm:n1}, \eqref{norm:n2}, \eqref{norm:n3}, \eqref{norm:n4}, \eqref{norm:n5}, \eqref{norm:n6}, \eqref{norm:n8}, \eqref{norm:n9}\big\},\\ &\big\{\eqref{norm:n2}, \eqref{norm:n3}, \eqref{norm:n4}, \eqref{norm:n5}, \eqref{norm:n6}, \eqref{norm:n7}, \eqref{norm:n9}\big\}, \\ & \big\{\eqref{norm:n1}, \eqref{norm:n3}, \eqref{norm:n4}, \eqref{norm:n5}, \eqref{norm:n6}, \eqref{norm:n7}, \eqref{norm:n9}\big\} \Big\}
\end{split}\end{equation*}
for $C = A$.\footnote{In this case study only the case of $C = A$, i.e., output consistent with the situation ($A$), is assessed. Also, the underlying output operator is $\out_3$ for all queries. The outputs primarily rely on reusable output, so using $\out_4$ does not change the results; when using $\out_1$ or $\out_2$ however, some outputs are lost.}

Intuitively, the first set of norms in this maxfamily describes a possible point of view from the perspective of the ECJ;
the second and third those of the BVerfG. In the latter case, there exists a natural conflict between norm~\eqref{norm:n1} and
norm~\eqref{norm:n2}, resulting in two different sets for the BVerfG.
The skeptical output ${\out_3}_{\cap}^C(N,A)$ is then given by
\begin{equation*}\begin{split}
{\out_3}_{\cap}^C(N,A) = &\cn(\textit{assessmentOfProportionalityECJ}, \\ &\textit{interpretationECJ}, 
\neg\textit{applicationECJ}, \\
&\textit{methodicallySoundECJ} \Rightarrow \textit{BVerfGFollowsECJ})
\end{split}\end{equation*}
as witnessed by \rio's output (produced by the input problem displayed in Appendix~\ref{appendix:file2}):
\begin{Verbatim}[frame=single]
% SZS status Success
% SZS output start ListOfFormulae
  assessmentOfProportionalityECJ
  interpretationECJ
  ~applicationECJ
  ~methodicallySoundECJ | bVerfGFollowsECJ
% SZS output end ListOfFormulae
\end{Verbatim}
An interesting observation is that, according to the constrained output,
there is no obligation of the BVerfG to follow or not to follow the ECJ's decision.
It is only stated that the BVerfG must follow the ECJ if its decision is methodically sound.

This result is quite interesting as no clear solution of the decision problem (or a verification of the BVerfG's decision)
can be inferred.
From a logical perspective, of course, this is quite trivial to see: Every norm is treated equally wrt.\
speciality and/or preference, so usual legal principles such as \emph{lex specialis derogat legi generali}
(more specialized norms defeat general norms) are not respected.

Because of the incompatible view points, the credulous output is inconsistent,
${\out_3}_{\cup}^C(N,A) = \cn(\bot)$ , as verified by \rio\ (not shown here).

In order to accommodate a more fine-grained analysis of the case, and to assess the situation from
different view points, the formalization is now augmented with
a preference ordering between the norms in order to allow for \emph{preferred output}~\citep{DBLP:journals/ail/Parent11},
as briefly introduced further above.

Following \emph{lex specialis derogat legi generali} it seems meaningful to assume that
$\eqref{norm:n1} \succ \eqref{norm:n2}$.
Further assuming that the norms originating from the different viewpoints are equally preferred,
i.e., $\eqref{norm:n6} \simeq \eqref{norm:n7} \simeq \eqref{norm:n8} \simeq \eqref{norm:n9}$
where $\simeq$ denotes an equal preference between the involved norms, the preferred output becomes
(as witnessed by \rio, not shown here, but see Appendix~\ref{appendix:file3} for the problem file):
\begin{equation*}\begin{split}
  {\out_3}_{\cap}^{C,\succeq^*}(N,A) = &\cn(\textit{assessmentOfProportionalityECJ}, \\ &\textit{interpretationECJ}, 
\neg\textit{applicationECJ}, \\
&\textit{methodicallySoundECJ} \Leftrightarrow \textit{BVerfGFollowsECJ})
\end{split}\end{equation*}
This output strengthens the skeptical output (without preferences) in the sense that
the BVerfG must follow the ECJ if and only if the decision of the ECJ is methodically sound
(as opposed to a simple implication). This is because the two conflicting sets in the maxfamily that are ascribed to the BVerfG's point of view
are now reduced to one set due to the preference relation.
Still, it is not possible to conclude any argument for or against one of the two different positions.

In contrast, if it is assumed that
$\eqref{norm:n6} \simeq \eqref{norm:n7}$ and $\eqref{norm:n8} \simeq \eqref{norm:n9}$, while 
$\eqref{norm:n6} \succ \eqref{norm:n8}$, $\eqref{norm:n6} \succ \eqref{norm:n9}$, 
$\eqref{norm:n7} \succ \eqref{norm:n8}$, and $\eqref{norm:n7} \succ \eqref{norm:n9}$, 
the position of the BVerfG is preferred to that of the ECJ. Consequently, this gives \ldots
\begin{equation*}\begin{split}
  {\out_3}_{\cap}^{C,\succeq^*}(N,A) = &\cn(\textit{assessmentOfProportionalityECJ}, \\ &\textit{interpretationECJ}, 
\neg\textit{applicationECJ}, \\
&\neg\textit{methodicallySoundECJ}, \neg\textit{BVerfGFollowsECJ})
\end{split}\end{equation*}
In this case, an obligation is detached that the BVerfG must not follow the ECJ's decision, and also that
the work on the ECJ must be considered methodically unsound (see Appendix~\ref{appendix:file4} for the problem file).

On the other hand, if the position of the ECJ is preferred
($\eqref{norm:n8} \succ \eqref{norm:n6}$, $\eqref{norm:n9} \succ \eqref{norm:n6}$, 
$\eqref{norm:n8} \succ \eqref{norm:n7}$, and $\eqref{norm:n9} \succ \eqref{norm:n7}$), a different
result is produced:
\begin{equation*}\begin{split}
  {\out_3}_{\cap}^{C,\succeq^*}(N,A) = &\cn(\textit{assessmentOfProportionalityECJ}, \\ &\textit{interpretationECJ}, 
\neg\textit{applicationECJ}, \\
&\textit{methodicallySoundECJ}, \textit{BVerfGFollowsECJ})
\end{split}\end{equation*}
Here, according to the position of the ECJ, the BVerfG has to follow its ruling, and its work has to be
considered methodically sound (see Appendix~\ref{appendix:file5} for the problem file).

The preference relation is provided externally (e.g., by the user of \rio\ or by argumentation systems) so that a
conclusive answer of the question (whether the BVerfG should have followed the ECJ's decision or not) remains
unanswered. Still, it is shown that \rio\ can be used to assess different normative setups to allow for
a computer-assisted analysis of non-trivial cases.

Note that outputs such as 
$\neg \textit{methodicallySoundECJ} \lor \textit{bVerfGFollowsECJ}$ and
$\textit{methodicallySoundECJ} \Leftrightarrow \textit{BVerfGFollowsECJ}$, as generated above,
are not heads of any norm involved in the inference process. This further motivates the usage of the
\problem\ problem as opposed to the entailment problem for output discovery, in particular when
the set of norms grows or the structure of outputs are non-trivial (e.g., because of preference-based output).

\section{Conclusion \label{sec:conclusion}}

In this paper a reduction of unconstrained and constrained I/O logic reasoning
to a sequence of SAT problems is presented. As opposed to earlier work of the author that
focused on the entailment problem in I/O logics, this work addresses the \problem\ problem that
asks for a finite set of formulas, called a basis,
from which every generated output is entailed. The reduction is 
sound and complete with respect to I/O logic semantics.

For a given normative system
$N$ and input $A$, the basis of the output set needs to be computed only once.
Subsequently, conjectured outputs can be checked for entailment from the basis using classical entailment
procedures (e.g., SAT solvers). Advantages of this approach are the following:
\begin{itemize}
  \item It provides a uniform automation approach for unconstrained and constrained I/O logics (with and without throughout),
        as well as further extensions resp.\ generalizations (see below).
  \item It makes explicit the structure of the set of entailed formulas. This information can, e.g.,
  be used to construct an explanation of why some specific formula is entailed.
  \item Next to transparency reasons, the \problem\ problem allows for a principled assessment
  of obligations in force. In particular, in complex scenarios it may not be clear which concrete obligation
  to check for entailment.
\end{itemize}
A notable disadvantage is that, if interested in the entailment problem only, the proposed method 
comes with a higher time and space complexity than direct methods, and so will provide less efficient
means of automation.

Furthermore, using an generalized net output operator $\gout$
the presented reasoning framework is general enough to support extensions of constrained I/O logic reasoning,
in particular preference-based I/O logic reasoning. Also, other forms of norm set selection
and output aggregation can easily be built on top of the existing methods.

The reasoning procedures have been implemented as a new normative reasoning software
called \rio\ that connects to the TPTP standard for automated theorem proving systems.
For this purpose, a representation format for I/O logic reasoning based on the 
TPTP NXF language standard is presented. The implementation is applied to a motivating case study,
illustrating the automation provided by \rio.
The procedures and their implementation in \rio\ constitute a building block towards (semi-)automated
normative reasoning systems, e.g., for compliance checking and
legal assessment.

\paragraph*{Further Work.}
An extensive practical evaluation of the approach and implementation is further work. This is because
there do not exist other automation systems for I/O logics, and only few small benchmarks. The build-up
of a representative collection of benchmark problems for I/O logic reasoning is a medium- and long-term goal.

Proof reconstruction for unconstrained I/O logic reasoning can be achieved by storing the set of (weakly) triggered set of norms, and the associated set of inputs,
for each individual output as meta information during the procedure. This is mainly engineering work, but will enable independent (external) proof verification
of results produced by \rio\ based on the sound and complete derivation rules of Makinson and van der Torre (\citeyear{DBLP:journals/jphil/MakinsonT00}).

I/O logics have been employed in the context of studying
conditional permissions~\citep{DBLP:journals/jphil/MakinsonT03} and it seems fruitful to integrate
both obligations and permissions within one automated reasoning framework.  
Recent work addresses weaker notions of I/O logic that allow for a fined-grained control over
employed inference principles~\citep{DBLP:conf/deon/ParentT18}. An integration of these techniques
remains further work.

The SAT-based approach seems fit to allow automated reasoning using so-called I/O logics without weakening~\citep{DBLP:conf/deon/ParentT14},
a generalization of the output operators discussed here. However, an implementation is not available yet.

I/O logics based on first-order languages have been studied in the context of
legal knowledge bases for privacy analysis with respect to the GDPR~\citep{robaldo2019formalizing}. 
Automating such a quantified language in the I/O framework for employment in practical normative reasoning
is the main motivation of this work. As a consequence, the next step is to generalize the SAT-based
encoding towards first-order reasoning using Satisfiability Modulo Theory (SMT) solvers~\citep{DBLP:series/faia/BarrettSST21}
instead of SAT solvers.
In fact, the MUS enumeration tool used by \rio\ already supports MUS extraction from SMT solvers.

\section*{Acknowledgments}
The author would like to thank the anonymous reviewers for their helpful comments and improvements.

\section*{Disclosure statement}
The author report there are no competing interests to declare.

\section*{Data availability statement}
The rio software version 1.3 is available at Zenodo via \url{https://doi.org/10.5281/zenodo.18757637}.

\bibliographystyle{apacite}
\bibliography{main}

\newpage
\appendix
\section{Input Problem Unconstrained Output} \label{appendix:file1}

\begin{Verbatim}[frame=single,fontsize=\footnotesize]
% I/O logic semantics specification
tff(semantics, logic, $$iol == [ $$operator == $$out3 ] ).

% Declaration of norms
tff(n1, axiom, {$$norm} @ (ultraViresECJ, ~bVerfGFollowsECJ) ).
tff(n2, axiom, {$$norm} @ ($true, bVerfGFollowsECJ) ).
tff(n3, axiom, {$$norm} @ (prelimRulingProcECJ, interpretationECJ) ).
tff(n4, axiom, {$$norm} @ (prelimRulingProcECJ, ~applicationECJ) ).
tff(n5, axiom, {$$norm} @ (prelimRulingProcECJ, assessmentOfProportionalityECJ) ).
tff(n9, axiom, {$$norm} @ (prelimRulingProcECJ & assessmentOfAdequacyECJ,
                           methodicallySoundECJ) ).
tff(n8, axiom, {$$norm} @ (prelimRulingProcECJ & ~assessmentOfAdequacyECJ &
                             ecbBondBuyingDecision, methodicallySoundECJ) ).
tff(n6, axiom, {$$norm} @ (prelimRulingProcECJ & assessmentOfAdequacyECJ,
                           methodicallySoundECJ) ).
tff(n7, axiom, {$$norm} @ (prelimRulingProcECJ & ~assessmentOfAdequacyECJ &
                             ecbBondBuyingDecision, ~methodicallySoundECJ) ). 

% Declaration of inputs
tff(a61, hypothesis, prelimRulingProcECJ).
tff(a71, hypothesis, assessmentOfProportionalityECJ).
tff(a72, hypothesis, ~assessmentOfAdequacyECJ).
tff(a62, hypothesis, ecbBondBuyingDecision).
tff(a2, hypothesis, interpretationECJ => abstractAssessmentECJ).
tff(a3, hypothesis, applicationECJ => concreteAssessmentECJ).
tff(a1, hypothesis, ~(abstractAssessmentECJ <=> concreteAssessmentECJ)).
tff(a4, hypothesis, assessmentOfAdequacyECJ => ~abstractAssessmentECJ).
tff(a5, hypothesis, ~methodicallySoundECJ => ultraViresECJ).
\end{Verbatim}

\newpage
\section{Input Problem Constrained Output} \label{appendix:file2}
\begin{Verbatim}[frame=single,fontsize=\footnotesize]
% I/O logic semantics specification
tff(semantics, logic, $$iol == [ $$operator == $$out3, 
                                 $$constrained == $$skeptical,
                                 $$constraints == $$input ] ).

% Declaration of norms
tff(n1, axiom, {$$norm} @ (ultraViresECJ, ~bVerfGFollowsECJ) ).
tff(n2, axiom, {$$norm} @ ($true, bVerfGFollowsECJ) ).
tff(n3, axiom, {$$norm} @ (prelimRulingProcECJ, interpretationECJ) ).
tff(n4, axiom, {$$norm} @ (prelimRulingProcECJ, ~applicationECJ) ).
tff(n5, axiom, {$$norm} @ (prelimRulingProcECJ, assessmentOfProportionalityECJ) ).
tff(n9, axiom, {$$norm} @ (prelimRulingProcECJ & assessmentOfAdequacyECJ,
                           methodicallySoundECJ) ).
tff(n8, axiom, {$$norm} @ (prelimRulingProcECJ & ~assessmentOfAdequacyECJ &
                             ecbBondBuyingDecision, methodicallySoundECJ) ).
tff(n6, axiom, {$$norm} @ (prelimRulingProcECJ & assessmentOfAdequacyECJ,
                           methodicallySoundECJ) ).
tff(n7, axiom, {$$norm} @ (prelimRulingProcECJ & ~assessmentOfAdequacyECJ &
                             ecbBondBuyingDecision, ~methodicallySoundECJ) ). 

% Declaration of inputs
tff(a61, hypothesis, prelimRulingProcECJ).
tff(a71, hypothesis, assessmentOfProportionalityECJ).
tff(a72, hypothesis, ~assessmentOfAdequacyECJ).
tff(a62, hypothesis, ecbBondBuyingDecision).
tff(a2, hypothesis, interpretationECJ => abstractAssessmentECJ).
tff(a3, hypothesis, applicationECJ => concreteAssessmentECJ).
tff(a1, hypothesis, ~(abstractAssessmentECJ <=> concreteAssessmentECJ)).
tff(a4, hypothesis, assessmentOfAdequacyECJ => ~abstractAssessmentECJ).
tff(a5, hypothesis, ~methodicallySoundECJ => ultraViresECJ).
\end{Verbatim}
\newpage

\section{Input Problem Preferred Output \#1} \label{appendix:file3}
\begin{Verbatim}[frame=single,fontsize=\footnotesize]
% I/O logic semantics specification
tff(semantics, logic, $$iol == [ $$operator == $$out3, 
                                 $$constrained == $$skeptical,
                                 $$preference == [[n6, n7, n8, n9], n1, n2],
                                 $$constraints == $$input ] ).

% Declaration of norms
tff(n1, axiom, {$$norm} @ (ultraViresECJ, ~bVerfGFollowsECJ) ).
tff(n2, axiom, {$$norm} @ ($true, bVerfGFollowsECJ) ).
tff(n3, axiom, {$$norm} @ (prelimRulingProcECJ, interpretationECJ) ).
tff(n4, axiom, {$$norm} @ (prelimRulingProcECJ, ~applicationECJ) ).
tff(n5, axiom, {$$norm} @ (prelimRulingProcECJ, assessmentOfProportionalityECJ) ).
tff(n9, axiom, {$$norm} @ (prelimRulingProcECJ & assessmentOfAdequacyECJ,
                           methodicallySoundECJ) ).
tff(n8, axiom, {$$norm} @ (prelimRulingProcECJ & ~assessmentOfAdequacyECJ &
                             ecbBondBuyingDecision, methodicallySoundECJ) ).
tff(n6, axiom, {$$norm} @ (prelimRulingProcECJ & assessmentOfAdequacyECJ,
                           methodicallySoundECJ) ).
tff(n7, axiom, {$$norm} @ (prelimRulingProcECJ & ~assessmentOfAdequacyECJ &
                             ecbBondBuyingDecision, ~methodicallySoundECJ) ). 

% Declaration of inputs
tff(a61, hypothesis, prelimRulingProcECJ).
tff(a71, hypothesis, assessmentOfProportionalityECJ).
tff(a72, hypothesis, ~assessmentOfAdequacyECJ).
tff(a62, hypothesis, ecbBondBuyingDecision).
tff(a2, hypothesis, interpretationECJ => abstractAssessmentECJ).
tff(a3, hypothesis, applicationECJ => concreteAssessmentECJ).
tff(a1, hypothesis, ~(abstractAssessmentECJ <=> concreteAssessmentECJ)).
tff(a4, hypothesis, assessmentOfAdequacyECJ => ~abstractAssessmentECJ).
tff(a5, hypothesis, ~methodicallySoundECJ => ultraViresECJ).
\end{Verbatim}
\newpage

\section{Input Problem Preferred Output \#2} \label{appendix:file4}
\begin{Verbatim}[frame=single,fontsize=\footnotesize]
% I/O logic semantics specification
tff(semantics, logic, $$iol == [ $$operator == $$out3, 
                                 $$constrained == $$skeptical,
                                 $$preference == [[n6, n7], [n8, n9], n1, n2],
                                 $$constraints == $$input ] ).

% Declaration of norms
tff(n1, axiom, {$$norm} @ (ultraViresECJ, ~bVerfGFollowsECJ) ).
tff(n2, axiom, {$$norm} @ ($true, bVerfGFollowsECJ) ).
tff(n3, axiom, {$$norm} @ (prelimRulingProcECJ, interpretationECJ) ).
tff(n4, axiom, {$$norm} @ (prelimRulingProcECJ, ~applicationECJ) ).
tff(n5, axiom, {$$norm} @ (prelimRulingProcECJ, assessmentOfProportionalityECJ) ).
tff(n9, axiom, {$$norm} @ (prelimRulingProcECJ & assessmentOfAdequacyECJ,
                           methodicallySoundECJ) ).
tff(n8, axiom, {$$norm} @ (prelimRulingProcECJ & ~assessmentOfAdequacyECJ &
                             ecbBondBuyingDecision, methodicallySoundECJ) ).
tff(n6, axiom, {$$norm} @ (prelimRulingProcECJ & assessmentOfAdequacyECJ,
                           methodicallySoundECJ) ).
tff(n7, axiom, {$$norm} @ (prelimRulingProcECJ & ~assessmentOfAdequacyECJ &
                             ecbBondBuyingDecision, ~methodicallySoundECJ) ). 

% Declaration of inputs
tff(a61, hypothesis, prelimRulingProcECJ).
tff(a71, hypothesis, assessmentOfProportionalityECJ).
tff(a72, hypothesis, ~assessmentOfAdequacyECJ).
tff(a62, hypothesis, ecbBondBuyingDecision).
tff(a2, hypothesis, interpretationECJ => abstractAssessmentECJ).
tff(a3, hypothesis, applicationECJ => concreteAssessmentECJ).
tff(a1, hypothesis, ~(abstractAssessmentECJ <=> concreteAssessmentECJ)).
tff(a4, hypothesis, assessmentOfAdequacyECJ => ~abstractAssessmentECJ).
tff(a5, hypothesis, ~methodicallySoundECJ => ultraViresECJ).
\end{Verbatim}
\newpage
\section{Input Problem Preferred Output \#3} \label{appendix:file5}
\begin{Verbatim}[frame=single,fontsize=\footnotesize]
% I/O logic semantics specification
tff(semantics, logic, $$iol == [ $$operator == $$out3, 
                                 $$constrained == $$skeptical,
                                 $$preference == [[n8, n9], [n6, n7], n1, n2],
                                 $$constraints == $$input ] ).

% Declaration of norms
tff(n1, axiom, {$$norm} @ (ultraViresECJ, ~bVerfGFollowsECJ) ).
tff(n2, axiom, {$$norm} @ ($true, bVerfGFollowsECJ) ).
tff(n3, axiom, {$$norm} @ (prelimRulingProcECJ, interpretationECJ) ).
tff(n4, axiom, {$$norm} @ (prelimRulingProcECJ, ~applicationECJ) ).
tff(n5, axiom, {$$norm} @ (prelimRulingProcECJ, assessmentOfProportionalityECJ) ).
tff(n9, axiom, {$$norm} @ (prelimRulingProcECJ & assessmentOfAdequacyECJ,
                           methodicallySoundECJ) ).
tff(n8, axiom, {$$norm} @ (prelimRulingProcECJ & ~assessmentOfAdequacyECJ &
                             ecbBondBuyingDecision, methodicallySoundECJ) ).
tff(n6, axiom, {$$norm} @ (prelimRulingProcECJ & assessmentOfAdequacyECJ,
                           methodicallySoundECJ) ).
tff(n7, axiom, {$$norm} @ (prelimRulingProcECJ & ~assessmentOfAdequacyECJ &
                             ecbBondBuyingDecision, ~methodicallySoundECJ) ). 

% Declaration of inputs
tff(a61, hypothesis, prelimRulingProcECJ).
tff(a71, hypothesis, assessmentOfProportionalityECJ).
tff(a72, hypothesis, ~assessmentOfAdequacyECJ).
tff(a62, hypothesis, ecbBondBuyingDecision).
tff(a2, hypothesis, interpretationECJ => abstractAssessmentECJ).
tff(a3, hypothesis, applicationECJ => concreteAssessmentECJ).
tff(a1, hypothesis, ~(abstractAssessmentECJ <=> concreteAssessmentECJ)).
tff(a4, hypothesis, assessmentOfAdequacyECJ => ~abstractAssessmentECJ).
tff(a5, hypothesis, ~methodicallySoundECJ => ultraViresECJ).
\end{Verbatim}

\end{document}